\documentclass[12pt]{article}
\usepackage{bm}
\usepackage{amsmath}
\usepackage{amssymb}
\usepackage{amsthm}
\usepackage{xcolor}
\usepackage{dsfont}
\usepackage{graphicx}
\usepackage{amsfonts}
\usepackage{bm}
\usepackage{mathtools}
\usepackage[sectionbib]{natbib}
\usepackage[english]{babel}
\usepackage{hyperref}
\usepackage{geometry}
\usepackage{bbm}
\usepackage{subfigure}
\usepackage{diagbox}
\usepackage{float}
\usepackage[linesnumbered,vlined,ruled]{algorithm2e}
 \usepackage{caption}

\theoremstyle{plain}
\newtheorem{theorem}{Theorem}[section]
\newtheorem{proposition}[theorem]{Proposition}
\newtheorem{lemma}[theorem]{Lemma}

\theoremstyle{definition}
\newtheorem{definition}[theorem]{Definition}

\newtheorem{condition}{Condition}
\newtheorem{condition2}{KRR Condition}
\theoremstyle{remark}
\newtheorem{remark}[theorem]{Remark}

\definecolor{myblue}{RGB}{0,102,153}

\newcommand{\BPhi}{\boldsymbol \Phi}

\newcommand{\mc}[1]{\mathcal{#1}}

\newcommand{\abs}[1]{\left|#1\right|}
\newcommand{\normone}[1]{\lVert #1 \rVert_1}

\newcommand{\wal}{\ensuremath{\mathrm{wal}}}

\usepackage{bm}
\newcommand{\boldell}{\bm{\ell}}

\newcommand{\bl}{\bm{\ell}}
\newcommand{\by}{\bm{y}}
\newcommand{\bk}{\bm{k}}
\newcommand{\bt}{\bm{t}}
\newcommand{\ba}{\bm{a}}
\newcommand{\bh}{\bm{h}}

\newcommand{\mcZ}{\mc{Z}}

\newcommand{\MHlprime}{M(H_{\boldsymbol{l}^{\prime}})}

\newcommand{\MprimeHlprime}{M^\prime(H_{\boldsymbol{l}^{\prime}})}

\newcommand{\Hlprime}{H_{\boldell^\prime}}

\newcommand{\floor}[1]{\left\lfloor #1 \right\rfloor}

\newcommand{\cardu}{|u|}

\newcommand{\bomega}{\bm{\omega}}

\usepackage{csquotes}
\newcommand{\dd}{\,\mathrm{d}}

\usepackage{xcolor}

\usepackage{xcolor} %

\usepackage{xcolor} %
\definecolor{deepred}{RGB}{139, 0, 0} %

\newcommand{\R}{\mathbb{R}}
\newcommand{\x}{\mathbf{x}}

\newcommand{\X}{\mathbf{X}}
\newcommand{\cX}{\mathcal{X}}

\newcommand{\E}{\mathbb{E}}

\newcommand{\y}{\mathbf{y}}

\newcommand{\pph}{\pmb{\phi}}

\newcommand{\tildeC}{\tilde{C}}

\usepackage{authblk}

\title{Randomized Quasi-Monte Carlo Features for Kernel Approximation}

\author[1]{Yian Huang\thanks{yh3209@columbia.edu}\textsuperscript{\ddag}}
\author[1]{Zhen Huang\thanks{zh2395@columbia.edu}\textsuperscript{\ddag}}

\affil[1]{Department of Statistics, Columbia University}

\date{\today} %

\begin{document}
\maketitle %

\renewcommand{\thefootnote}{\fnsymbol{footnote}}
\footnotetext[3]{Equal contribution}
\renewcommand{\thefootnote}{\arabic{footnote}}

\begin{abstract}
We investigate the application of randomized quasi-Monte Carlo (RQMC) methods in random feature approximations for kernel-based learning. Compared to the classical Monte Carlo (MC) approach \citep{rahimi2007random}, RQMC improves the deterministic approximation error bound from $O_P(1/\sqrt{M})$ to
$O(1/M)$ (up to logarithmic factors), matching the rate achieved by quasi-Monte Carlo (QMC) methods \citep{huangquasi}. Beyond the deterministic error bound guarantee, we further establish additional average error bounds for RQMC features: some requiring weaker assumptions and others significantly reducing the exponent of the logarithmic factor. In the context of kernel ridge regression, we show that RQMC features offer computational advantages over MC features while preserving the same statistical error rate. Empirical results further show that RQMC methods maintain stable performance in both low and moderately high-dimensional settings, unlike QMC methods, which suffer from significant performance degradation as dimension increases.
\end{abstract}

\section{Introduction}

Kernel methods constitute a fundamental class of techniques  in machine learning, offering powerful tools for tackling complex nonparametric estimation and inference tasks in classification, regression, and beyond. Their flexibility and theoretical guarantees have led to widespread use in practice~\citep{scholkopf2002learning, evgeniou2005learning, Hofmann2007KernelMI,muller2018introduction}.
Despite these strengths, kernel methods often suffer from prohibitive computational costs.
To mitigate these scaling challenges, various approaches have been developed, notably low-rank approximations and randomized feature mappings~\citep{williams2001using, rahimi2007random}. Among these, Monte Carlo (MC) random features for kernel approximation~\citep{rahimi2007random} have gained significant popularity, as they are easy to implement and drastically reduce the computational complexity in kernel-based learning~\citep{liu2022survey, sinha2016learning, bach2017equivalence, chen2022online}.

While the MC random feature approach has proven effective, it inherently suffers from the statistical limitations of random sampling. In recent years, attention has turned to improving the accuracy and stability of kernel approximations by employing more carefully designed sampling schemes. Quasi–Monte Carlo (QMC) methods~\citep{niederreiter1992random,dick2010digital}, which replace purely random sampling with low-discrepancy point sets, have been shown to yield more accurate approximations under certain conditions. %
Some studies have successfully applied QMC features for kernel approximation~\citep{yang2014quasi,avron2016quasi,huangquasi}, demonstrating improvements over standard MC random features in low-dimensional settings. However, these benefits are reported to degrade with increasing dimension, typically becoming negligible or even detrimental when the dimension exceeds 10~\citep{huangquasi}.

To address the poor scalability of QMC features in higher dimensions, we explore a new approach: randomized quasi–Monte Carlo (RQMC) features for kernel approximation. 
By incorporating appropriate randomization schemes, RQMC can maintain the improved convergence properties of QMC in low-dimensional scenarios without succumbing to the curse of dimensionality that plagues non-randomized QMC methods~\citep{l2018randomized,hok2022importance}. 
In this paper, we establish the average and deterministic error bounds for RQMC based kernel approximation, and the theoretical guarantees for application of the RQMC method in kernel ridge regression (KRR).
Our proposed RQMC-based kernel approximation method proves to be computationally more efficient than MC methods, and
does not incur any theoretical loss in the statistical error rate.

We further support these theoretical insights with comprehensive empirical evaluations. In low-dimensional settings, our experiments confirm that the proposed RQMC features  match the performance of QMC features, offering significantly improved accuracy over the MC method with the same number of features. In intermediate to moderately high dimensional domains, RQMC does not degrade as QMC does. Instead, it remains competitive, typically exhibiting performance on par with MC random features or better, and consistently outperforming the non-randomized QMC approach. 
Hence, RQMC based kernel approximation  emerges as a robust alternative that seamlessly adapts to varying dimensional complexities without sacrificing theoretical soundness or empirical stability, offering a valuable new tool for scalable kernel methods in modern machine learning applications.

\subsection{Background on QMC and RQMC}
We first introduce the necessary background on QMC and RQMC. 
For a comprehensive introduction, we refer readers to \citet{niederreiter1992random,dick2010digital,owen2023practical}.
QMC methods replace random samples used in MC with carefully constructed deterministic sequences, often referred to as {low-discrepancy (LD) sequences}. 
The key property of these sequences is their low {\it star discrepancy},
where the {star discrepancy} $D^*_{M}$ of a sequence $\{\mathbf{x}_1, \ldots, \mathbf{x}_M\}$ in $[0,1]^d$ is defined as
\begin{equation*}
  D^*_{M} \;=\; \sup_{\mathbf{t} \in [0,1]^d} \Bigg| \frac{1}{M}\sum_{n=1}^{M} \mathbf{1}\{\mathbf{x}_n \le \mathbf{t}\} \;-\; \prod_{j=1}^d t_j \Bigg|,
\end{equation*}
where $\mathbf{x}_n \le \mathbf{t}$ means $x_{n,j} \le t_j$ for each dimension $j = 1,\ldots,d$.
Intuitively, the star discrepancy captures how much deviation there is between the empirical distribution of the points and the ideal uniform distribution.
LD sequences satisfy 
\[
D^*_{M} = O\bigl((\log M)^d / M\bigr),
\]
which is superior to the $O_P(M^{-1/2})$ star discrepancy decay of classic MC~\citep{dick2010digital,owen2023practical}.

\paragraph{Koksma--Hlawka Inequality.}
The Koksma-Hlawka (KH) inequality serves as the theoretical foundation for the QMC approach for approximating integrals.
Consider the integral
\begin{equation*}
  I(f) \;=\; \int_{[0,1]^d} f(\mathbf{x}) \, \dd \mathbf{x}.
\end{equation*}
The Koksma--Hlawka inequality \citep{hlawka1961funktionen,niederreiter1992random} states that for a function $f$ of \emph{bounded variation} $V(f)$ in the sense of Hardy and Krause,
\begin{equation*}
  \Bigg|\frac{1}{M}\sum_{n=1}^{M} f(\mathbf{x}_n) \;-\; I(f)\Bigg| \;\le\; V(f)\, D^*_{M},
\end{equation*}
where $\{\mathbf{x}_n\}_{n=1}^M$ is a sequence of points in $[0,1]^d$, and $D^*_{M}$ is its star discrepancy. Hence, the rate of convergence of the QMC estimator depends on both the smoothness of $f$ and the discrepancy of the sequence used for sampling.

Some commonly used QMC sequences include Halton, Sobol' and Faure sequences~\citep{halton1960efficiency,halton1964algorithm,sobol1967distribution,faure1982discrepance}.
Halton sequences \citep{halton1960efficiency} generalize the one-dimensional Van der Corput sequence to multiple dimensions by using distinct prime bases for each dimension. 
Sobol' sequences \citep{sobol1967distribution} are among the most popular LD sequences in machine learning and statistical contexts due to their ease of generation and good empirical performance. They are constructed based on direction numbers in base $2$ and exhibit low star discrepancy, and can be combined with randomization strategies like {scrambling} \citep{owen1997scrambled} to construct RQMC sequences.
Faure sequences \citep{faure1982discrepance} are another class of LD sequences constructed in base $b$ using a permuted polynomial representation. 

Digital nets constitute a broader framework for generating low-discrepancy point sets in $[0,1]^d$ \citep{niederreiter1992random, dick2010digital}. 
Let $d \ge 1$ and $b \ge 2$ be integers. An \emph{elementary interval in base ${b}$} is a
subinterval of $[0, 1)^d$ of the form
$$
E_{\bm{k},\bm{c}} = \prod_{j=1}^d \left[ \frac{c_j}{b^{k_j}}, \frac{c_j + 1}{b^{k_j}} \right)
$$
for integers $k_j$ and $c_j$, with $k_j \ge 0$ and $0 \le c_j < b^{k_j}$. 
Let $m \ge t \ge 0$ be integers. The sequence $\boldsymbol{x}_1, \dots, \boldsymbol{x}_{b^m} \in [0, 1)^d$
is a $(t, m, d)$-net in base ${b}$ if every elementary interval in base $b$ of volume
$b^{t-m}$ contains exactly $b^t$ points of the sequence. Intuitively, this means that every subregion of the space (in the form of an elementary interval) gets a fair share of points. 
The infinite extensions of the $(t,m,d)$-nets are called $(t,d)$-sequences.
For $t \geq 0$, the infinite sequence $\x_1, \x_2, \ldots \in [0, 1]^d$ is a $(t, d)$-sequence in base $b$ if for all $k \geq 0$ and $m \geq t$ the sequence $\x_{k b^m + 1}, \ldots, \x_{(k+1) b^m}$ is a $(t, m, d)$-net in base $b$~\citep{owen2023practical}.
In particular, Faure sequences are $(0,d)$-sequences in base $p$ with $p \ge d$ being a prime number,
and Sobol' sequences are $(t,d)$-sequences in base $2$~\citep{owen2023practical}.

Randomized Quasi-Monte Carlo (RQMC) \citep{owen1995randomly,owen1997scrambled, l2002recent} methods aim to combine the best of both worlds: they preserve the low-discrepancy structure of QMC sequences while incorporating a layer of randomness that enables unbiased estimation and variance evaluation through replication. In one widely used RQMC technique, \emph{digital scrambling}~\citep{owen1995randomly}, each point of a QMC sequence is randomly permuted in its digital representation in base $b$, yielding a {scrambled} sequence whose discrepancy properties remain advantageous while still allowing the practitioner to compute empirical variances in a manner similar to standard MC. 
The practical utility of RQMC is especially pronounced in high-dimensional integration tasks where direct QMC can still face challenges, but well-chosen scramblings can often noticeably reduce the variance compared to classic MC.

The digital nets mentioned earlier can be scrambled to obtain scrambled nets while preserving the low discrepancy features. For instance, one can apply a random linear transformation or random permutation of the digits in the base-$b$ expansions of each point. The resulting scrambled digital net inherits the low-discrepancy characteristics of the original net, allowing the construction of unbiased RQMC estimators \citep{owen1997scrambled, l2002recent}. In particular, a scrambled $(t, m, d)$-net remains a $(t, m, d)$-net with probability 1 after scrambling 
\citep[Proposition 17.2]{owen2023practical}. Scrambling not only provides a mechanism for variance estimation but also mitigates certain systematic artifacts that can arise when employing purely deterministic point sets, especially in higher dimensions.

The implementations of RQMC sequences are available in major computational
softwares (e.g., the Python package SciPy \citealp{2020SciPy-NMeth}).

\subsection{Literature Review}

Kernel methods underpin many machine learning algorithms, including kernel ridge regression, support vector machines and Gaussian processes, by allowing nonlinear decision functions to be learned efficiently in a high-dimensional Reproducing Kernel Hilbert Space (RKHS) \citep{scholkopf2002learning, Cortes1995, Rasmussen2006, huang2022kernel,gretton2012kernel,belkin2006manifold}. 
However, their high computational complexity often hinders scalability~\citep{rudi2017generalization, lu2014scale, cesa2015complexity}. 
A significant advance to mitigate this issue is the use of {random features}, where the kernel function is represented by an inner product in a finite-dimensional space \citep{rahimi2007random}.
Other variants have also been investigated, such as using structured transforms \citep{Le2013} or refining the sampling strategy \citep{Sutherland2015}.
For KRR, it is shown that RF can achieve significant reduction in computational complexity without sacrificing statistical accuracy~\citep{li2019towards, rudi2017generalization, avron2017random}.
Despite these advances, standard MC sampling remains vulnerable to high variance and slow convergence. 

Quasi-Monte Carlo (QMC) methods, introduced in~\cite{korobov1959approximate, korobov1963number}, 
using low-discrepancy sequences, 
offer a means of more uniformly covering the input space, potentially improving estimation quality and convergence \citep{niederreiter1992random, dick2010digital, dick2013high}.

 While QMC sequences show promise, they are not without drawbacks, and could suffer in higher dimensional settings \citep{huangquasi}. 
 RQMC methods, introduced in~\cite{cranley1976randomization, owen1995randomly, owen1997scrambled},
 provide a remedy by introducing a controlled form of randomness to QMC sequences, thereby producing randomized low-discrepancy samples~\citep{owen2023practical, dick2010digital}.

QMC and RQMC have been used in kernel methods in the literature. 
\citet{ben2021density} studied the effectiveness of RQMC for kernel density estimation.
\cite{di2022quasi} examined the use of QMC for exponentiated quadratic kernel in latent force models.
\cite{hertrich2024fast} used QMC slicing for fast summation of radial kernels.
\cite{yang2014quasi, avron2016quasi} used QMC sequences to enhance the efficiency of RF  and introduced a discrepancy measure called 
\enquote{box discrepancy}.
\cite{huangquasi} 
further solidified the theoretical foundation of using QMC for kernel approximation; they demonstrated that for a broad class of kernels, including the widely used Gaussian kernel, QMC methods can achieve a significant improvement in approximation error.
They also highlighted the benefits of QMC features in kernel ridge regression, where fewer random features are needed to achieve the same level of accuracy.
However, higher dimensional challenges exist for the QMC method: when the dimension exceeds roughly $10$, 
the performance of QMC features was observed to degrades significantly and becomes even worse than the vanilla MC method~\citep{huangquasi}.
To address this problem, we propose the use of RQMC, and establish the average and deterministic error bounds for RQMC features in kernel approximation, as well as the theoretical guarantees for its performance in KRR.

\subsection{Organization}
We present the RQMC based kernel approximation approach in Section~\ref{sec:approx_kernel_func_with_rqmc}, and provide both average-case and deterministic-case approximation error bounds. In Section~\ref{sec:KRR}, we show that in the application of kernel ridge regression, the RQMC-based random features achieve the same statistical error rate as the exact KRR, with lower computational cost compared with MC-based random features. 
In Section~\ref{sec:simulation}, the empirical evidence is provided to show that 
in kernel approximation and KRR, while QMC based random features degrades significantly as the dimension increases, the RQMC based random features remain stable in both low dimensions and moderately high dimensions, making it a preferred choice in practice.
Proofs and additional simulation results are presented in the appendices.

\section{Approximate Kernel Functions with RQMC}
\label{sec:approx_kernel_func_with_rqmc}
In this section, we introduce RQMC-based kernel approximation, and provide the average-case and deterministic-case error bounds.

Kernel methods often rely on a kernel function $K: \mathcal{X} \times \mathcal{X} \to \mathbb{R}$, where $\mathcal{X} \subseteq \mathbb{R}^d$, that admits an integral representation of the form
\begin{equation}
\label{eq:integralrepresentation}
K\bigl(\mathbf{x}, \mathbf{x}^{\prime}\bigr)
= \int_{\Omega} \psi\bigl(\mathbf{x}, \omega\bigr)\,\psi\bigl(\mathbf{x}^{\prime}, \omega\bigr)\, \mathrm{d}\pi(\omega),
\end{equation}
where $\pi$ is a probability measure defined over some space $\Omega$, and $\psi(\cdot, \cdot)$ is a suitable mapping from $\mathcal{X} \times \Omega$ to $\mathbb{R}$. A notable instance of such kernels arises when $K$ is shift-invariant, i.e., $K\bigl(\mathbf{x}, \mathbf{x}^{\prime}\bigr) = h\bigl(\mathbf{x} - \mathbf{x}^{\prime}\bigr)$. The Bochner's theorem \citep{bochner1933monotone} states that every continuous, shift-invariant kernel on $\mathbb{R}^d$ is the Fourier transform of a finite non-negative symmetric Borel measure $\mu$ on $\mathbb{R}^d$, such that
\begin{equation}
\label{eq:bochner}
\begin{aligned}
h(\mathbf{x} - \mathbf{x}') 
&= \int_{\mathbb{R}^{d}} e^{-i(\mathbf{x} - \mathbf{x}')^\top \omega}\,\mathrm{d}\,\mu(\omega) \\
&= \int_{\mathbb{R}^d} \int_{0}^{2\pi}
   \frac{1}{\pi}\cos(\mathbf{x}^\top \omega + b)\,
   \cos(\mathbf{x}'^\top \omega + b)\,\mathrm{d}b\,\mathrm{d}\mu(\omega).
\end{aligned}
\end{equation}
This framework encompasses many commonly used kernels, such as Gaussian kernel, Laplacian kernel and Cauchy kernel~\citep{huangquasi}. For example, for Gaussian kernel $K(\x,\x')=\exp(-\|\sigma (\x - \x')\|^2/2)$, $\mu$ is the Gaussian measure $\mu\sim \mathcal{N}(\mathbf{0}, \sigma^2 \mathbf{I}_d)$.

When $K$ can be represented via the above integral \eqref{eq:bochner}, one may approximate $K\bigl(\mathbf{x}, \mathbf{x}^{\prime}\bigr)$ by an average:
\begin{equation}
K_{M}\bigl(\mathbf{x}, \mathbf{x}^{\prime}\bigr) 
= \frac{1}{M}\sum_{i=1}^{M} \psi\bigl(\mathbf{x}, \omega_i\bigr)\,\psi\bigl(\mathbf{x}^{\prime}, \omega_i\bigr),
\end{equation}
where $\{\omega_i\}_{i=1}^M$ are independent and identically distributed (i.i.d.) random variables drawn from $\pi$. This strategy forms the basis of the well-known random features approach in kernel methods, which reduces the computational complexity of the kernel ridge regression 
($O(n^3)$ in time; ${O}(n^2)$ in space) to that of the ordinary ridge regression on \(\mathbb{R}^M\) 
(\({O}(nM^2 + M^3)\) in time; \({O}(nM)\) in space), with \(M \ll n\).

Here, we propose to replace the MC sequence by a randomized quasi-Monte Carlo sequence:
\begin{definition}[RQMC features]
    Suppose there exists a function $\psi:\mathcal{X}\times [0,1]^p\to\R$ such that
$$K(\x,\x')=\int_{[0,1]^p}\psi(\x,\omega)\psi(\x',\omega) {\rm d} \omega.$$
The kernel $K\bigl(\mathbf{x}, \mathbf{x}^{\prime}\bigr)$ is approximated by RQMC features as follows:
\begin{equation}
\label{eq:K_M}
K_{M}\bigl(\mathbf{x}, \mathbf{x}^{\prime}\bigr) 
= \frac{1}{M}\sum_{i=1}^{M} \psi\bigl(\mathbf{x}, \omega_i\bigr)\,\psi\bigl(\mathbf{x}^{\prime}, \omega_i\bigr),
\end{equation}
where $\{\omega_i\}_{i=1}^M$ are a sequence of RQMC features.
\end{definition}

Assume that $\mu$ from Bochner's theorem~\eqref{eq:bochner} is a probability measure with independent components, with the $i$-th component having cumulative distribution function $\Phi_{i}(t)(i=1,2, \ldots, d)$. Let $\boldsymbol{\Phi}(\mathbf{t}):=\left(\Phi_{1}(\mathbf{t}), \ldots, \Phi_{d}(\mathbf{t})\right)^{\top}$, and $\boldsymbol{\Phi}^{-1}(\mathbf{t}):=\left(\Phi_{1}^{-1}(\mathbf{t}), \ldots, \Phi_{d}^{-1}(\mathbf{t})\right)^{\top}$, where $\Phi_{i}^{-1}(\mathbf{t})$ denotes the inverse function of the monotone function $\Phi_{i}(\mathbf{t})$. By a change of variable, \eqref{eq:bochner} reduces to
\begin{equation}
\begin{split}
\label{eq:shift-invariant-kernel-integral}
& K\left(\mathbf{x}, \mathbf{x}^{\prime}\right)=h\left(\mathbf{x}-\mathbf{x}^{\prime}\right)= \\
& \int_{[0,1]^{d+1}} 2 \cos \left(\mathbf{x}^{\top} \mathbf{\Phi}^{-1}(\mathbf{t})+2 \pi b\right) \cos \left(\left(\mathbf{x}^{\prime}\right)^{\top} \mathbf{\Phi}^{-1}(\mathbf{t})+2 \pi b\right) \mathrm{d} b \mathrm{~d} \mathbf{t} .
\end{split}
\end{equation}
Therefore, the integral representation \eqref{eq:integralrepresentation} holds with $\omega=(\mathbf{t}, b)$ following Unif $[0,1]^{d+1}$ and $\psi(\mathbf{x}, \omega)=$ $\sqrt{2} \cos \left(\mathbf{x}^{\top} \mathbf{\Phi}^{-1}(\mathbf{t})+2 \pi b\right)$.

As $t$ approaches the boundary of $[0,1]^{d}$, the integrand in~\eqref{eq:shift-invariant-kernel-integral} 
oscillates back and forth and has unbounded variation (so classical Koksma-Hlawka inequality is not applicable). We therefore need a condition to characterize the situation where the singularity is mild so that $K$ can still be well approximated by $K_{M}$. We adopt the regularity condition as in \citet{huangquasi}:

\begin{condition}
\label{cond1}
$K(\cdot, \cdot)$ is shift invariant with $\Phi_{i}$ defined as above $(i=1, \ldots, d)$ satisfying $\frac{\mathrm{d}}{\mathrm{d} t} \Phi_{i}^{-1}(t) \leq \frac{C_{i}}{\min (t, 1-t)}$ for some constant $C_{i}>0$ and all $t \in(0,1)$.
$\mathcal{X}$ is compact.
\end{condition}

Condition \ref{cond1} helps control the derivatives of the integrand in~\eqref{eq:shift-invariant-kernel-integral} as $\mathbf{t}$ approaches the boundary of $[0,1]^{d}$. It is known that the Gaussian kernel and Cauchy kernel over a compact domain satisfy Condition 1
\citep[Proposition 2.1]{huangquasi}.

\subsection{Average Error Bound}
In this section, we establish several average error bounds for the proposed RQMC features.

Under the Condition \ref{cond1}, we first establish an average-case approximation error bound of $K_M$ to $K$ below (see Appendix~\ref{pf:thm:qmc1_averagecase_kernelapprox_error} for a proof).

\begin{theorem}
\label{thm:qmc1_averagecase_kernelapprox_error}
    Suppose $K(\cdot, \cdot)$ satisfies Condition \ref{cond1}, and an RQMC sequence on $[0,1]^{d+1}$ satisfying $\mathcal{D}^{*}\left(\left\{\mathbf{h}_{i}\right\}_{i=1}^{M}\right) \leq$ $C \frac{\log ^{d+1} M}{M}$ for all $M \geq 2$ is used. Then there exists a constant $C^{\prime}>0$ (depending on $C, \mathcal{X} \subset \mathbb{R}^{d}$ and $K$) such that for all $M \geq 2$,
$$
\sup _{\mathbf{x}, \mathbf{x}^{\prime}}
\mathbb{E}
\left|K_{M}\left(\mathbf{x}, \mathbf{x}^{\prime}\right)-K\left(\mathbf{x}, \mathbf{x}^{\prime}\right)\right|
\leq C^{\prime} \frac{(\log M)^{2 d+1}}{M}.
$$
\end{theorem}
\begin{remark}
Compared with \citet[Theorem 2.2]{huangquasi}, Theorem~\ref{thm:qmc1_averagecase_kernelapprox_error} is an average error bound instead of a deterministic one. Note that any point from an RQMC sequence marginally follows the uniform distribution over the unit cube, which does not deterministically avoid the boundary. Therefore, the technique used in 
\citet[Theorem 2.2]{huangquasi}
cannot be directly applied. On the other hand, compared with \citet[Theorem 2.2]{huangquasi} for which the use of Halton sequence is crucial, Theorem~\ref{thm:qmc1_averagecase_kernelapprox_error} is not restricted to a particular choice of RQMC sequence.
\end{remark}
\begin{remark}
Here, we provide some examples of RQMC sequences for which the above Theorem 2.1 is applicable.
A widely used RQMC sequence is the scrambled Sobol' sequence, which has been well implemented in major computational softwares. For the scrambled Sobol' sequence, it is recommended to take $M$ as a power of 2 , for which the power of $(\log M)$ in the bound of $\mathcal{D}^{*}\left(\left\{\mathbf{h}_{i}\right\}_{i=1}^{M}\right)$ can be reduced by 1 \citep[Theorem 4.10]{niederreiter1992random}, i.e., for scrambled Sobol’ sequence $\left\{\mathbf{h}_{i}\right\}_{i=1}^{M}$ in dimension $d+1$,
    \begin{equation}
    \label{eq:star-discrepancy-bound}
        \mathcal{D}^{*}\left(\left\{\mathbf{h}_{i}\right\}_{i=1}^{M}\right) \leq C \frac{\log ^{d} M}{M}+O\left(\frac{\log ^{d-1} M}{M}\right).
    \end{equation}
    
If $\ba_1,\ldots,\ba_M$ is a $(t,m,d+1)$-net in base $b$, 
let $\bh_1,\ldots,\bh_M$ be 
a nested uniform  scramble of $\ba_1,\ldots,\ba_M$, then $\bh_1,\ldots,\bh_M$ is also a $(t,m,d+1)$-net in base $b$, with  probability 1~\citep{owen1995randomly}. Thus $\bh_1,\ldots,\bh_M$ also satisfy the star discrepancy bound in~\eqref{eq:star-discrepancy-bound}, according to~\citet[Theorem 4.10]{niederreiter1992random}.

Cranley-Patterson (CP) rotation~\citep{cranley1976randomization} provides another way to randomize QMC  points. It is shown that  low discrepancy points randomized by the CP rotation still has low discrepancy~\citep[(17.10)]{owen2023practical}. Therefore, an RQMC sequence resulting from a low discrepancy sequence (e.g., Halton sequence, $(t,m,d+1)$-nets) randomized by the CP rotation still satisfies the star discrepancy bound required in Theorem~\ref{thm:qmc1_averagecase_kernelapprox_error}.
\end{remark}

Theorem~\ref{thm:qmc1_averagecase_kernelapprox_error} focused on the averaged worst case error.
Next, we show that the averaged $L^2$ error of the RQMC estimator is also superior to that of the MC estimator.
We first introduce some notations:
A sequence $\left(\mathbf{X}_{i}\right)$ of $\lambda b^{m}$ points is called a $(\lambda, t, m, s)$-net 
\citep{owen1997monte}
in base $b$ if every elementary interval in base $b$ of volume $b^{t-m}$ contains $\lambda b^{t}$ points of the sequence and no elementary interval in base $b$ of volume $b^{t-m-1}$ contains more than $b^{t}$ points of the sequence. Here, $s, m, t, b, \lambda$ are integers with $s \geq 1$, $0 \leq t \leq m, b \geq 2$, and $1 \leq \lambda<b$. Trivially, a $(t, m, s)$-net in base $b$ is a $(1, t, m, s)$-net in base $b$, and for base $b=2$ all $(\lambda, t, m, s)$-nets are also $(t, m, s)$-nets. If $\left(\mathbf{X}_{i}\right)_{i \geq 1}$ is a $(t, s)$-sequence in base $b$, then $\left(\mathbf{X}_{i}\right)_{i=a b^{m+1}+1}^{a b^{m+1}+\lambda b^{m}}$ is a $(\lambda, t, m, s)$-net in base $b$ for integers $a \geq 0$ and $1 \leq \lambda<b$.

\begin{theorem}
\label{thm:kernel-variance}
    Let $K(\cdot, \cdot)$ be a shift-invariant kernel (or a non-shift invariant kernel with a square integrable integrand). Suppose an RQMC sequence on $[0,1]^{d+1}$ based on a scrambled $(\lambda, t, m, d+1)$-net with $m\ge t$ is used. Then for fixed $\mathbf{x}$ and $\mathbf{x}^{\prime}$, we have
$$
\mathbb{E}\left[\left|K_{M}\left(\mathbf{x}, \mathbf{x}^{\prime}\right)-K\left(\mathbf{x}, \mathbf{x}^{\prime}\right)\right|^{2}\right]=o(1 / M).
$$
\end{theorem}
\begin{remark}
    The above bound is an improvement compared with the MC method, where
$$
\mathbb{E}\left[\left|K_{M}\left(\mathbf{x}, \mathbf{x}^{\prime}\right)-K\left(\mathbf{x}, \mathbf{x}^{\prime}\right)\right|^{2}\right]=O(1 / M).$$ Theorem~\ref{thm:kernel-variance} follows from a direct application of 
\citet[Theorem 1]{owen1998scrambling}, 
as the integrand $f \in L^{2}[0,1]^{d+1}$.
In particular, it does not require $K(\cdot, \cdot)$ to satisfy Condition~\ref{cond1}.
\end{remark}

Under slightly stronger conditions, it can be shown that 
$\sup_{\mathbf{x},\mathbf{x}' \in \mc{X}} \mathbb{E} [ |K_M(\mathbf{x}, \mathbf{x}') - K(\mathbf{x}, \mathbf{x}')|^2 ]$
is of the order of $O(\frac{\log^d (M)}{M^2})$.
Here, we introduce the following smoothness condition of the integrand which is a revised version of the boundary growth condition 
proposed by~\citet{owen2006halton}. 
\begin{condition}
\label{assumption:owen_boundary_growth_condition}

    Suppose $f_{\x,\x'}(\bomega)$ is a square integrable real-valued function on $[0, 1]^{d+1}$,
     and the derivative $\frac{\partial^{{u}} f_{\x,\x'}}{\partial \bm{\omega}_{{u}}}$ exists on $[0, 1]^{d+1}$ for any 
     ${u} \subseteq \{1,2, \dotsc, {d+1}\}$ 
     and any $\x,\x' \in \mc{X}$. 
  There exists a constant $C>0$ and constants $A_j \geq 0$ for $ j \in \{ 1,2,\ldots, d+1 \}$  such that
	\begin{equation}
    \label{eq:boundary_growth_condition}
		\sup_{\x,\x' \in \mc{X}}
  \abs{\frac{\partial^{{u}} f_{\x,\x'}}{\partial \bm{\omega}_{{u} } }} \leq C \prod_{j=1}^{d+1} 
  \min(\omega_j, 1-\omega_j)^{- \mathbbm{1}_{j \in {u}} - A_j}
	\end{equation}
for all ${u} \subseteq \{1,2, \dotsc, {d+1}\}$. 
\end{condition}
\begin{remark}
    The set $u$ in Condition~\ref{assumption:owen_boundary_growth_condition} can be an empty set. When $u = \emptyset$, and $A_j=0$ for $j = 1,2, \ldots,d+1$, we adopt the convention that $0^0=1$.
\end{remark}
\begin{remark}
\label{rk:cond2}
    If $K(\cdot, \cdot)$ is a shift-invariant kernel satisfying Condition \ref{cond1},
    then its integral representation satisfies the Condition~\ref{assumption:owen_boundary_growth_condition}.
    In fact,
    let $\bm{w}=(\bm{t},b)$, then by \citet[Appendix B.1]{huangquasi},
the integrand function $f_{\x,\x'}$ 
can be re-written as
$$\begin{aligned}
f_{\x,\x'}(\mathbf{t},b) &= \cos\left( (\x-\x')^\top \BPhi^{-1}(\mathbf{t})  \right) -  \cos\left( (\x+\x')^\top \BPhi^{-1}(\mathbf{t}) + 4\pi b \right).
\end{aligned} $$
Let $D=\max_{\x,\mathbf{y}\in \mathcal{X},i\in\{1,\ldots,d\}} \{ |x_i-{y}_i|,|x_i+{y}_i|\}$. For any non-empty set $u\subset \{1,\ldots,d+1\}$ and $(\mathbf{t},b)\in(0,1)^{d+1}$, we have
$$\left|\partial^u f_{\x,\x'}
(\mathbf{t},b) \right| \leq 4\pi  D^{|u\backslash\{d+1\}|} \prod_{i\in u\backslash\{d+1\}}\frac{{\rm d} }{{\rm d}t_i}\Phi^{-1}_i(t_i).$$
And by Condition \ref{cond1}, $\frac{{\rm d} }{{\rm d}t}\Phi^{-1}_i(t)\leq \frac{C_i}{\min(t,1-t)} $ for some constant $C_i>0$ and all $ t\in (0,1)$. Therefore, the Condition \ref{assumption:owen_boundary_growth_condition} is satisfied with $C = 4\pi  D^{|u\backslash\{d+1\}|} \prod_{i\in u\backslash\{d+1\}} C_i$
and all $A_j=0$.
\end{remark}

Assuming $f_{\x,\x'}(\bomega)=\psi(\x,\bomega)\psi(\x',\bomega)$
satisfies Condition~\ref{assumption:owen_boundary_growth_condition}, we have
the following average error bound (see Appendix~\ref{pf:thm:sup_meansquare} for a proof).
\begin{theorem}
\label{thm:sup_meansquare}
    Let $\mc{X}$ be a bounded domain. Suppose $K(\cdot, \cdot)$ %
    is a shift-invariant kernel satisfying Condition \ref{cond1},
    or a general kernel with a square integrable integrand $f_{\x,\x'}(\bomega)=\psi(\x,\bomega)\psi(\x',\bomega)$
    satisfying 
    Condition \ref{assumption:owen_boundary_growth_condition} with all $A_j=0$. 
    Suppose the first $M=2^m \,(m \ge 4)$ points of a scrambled Sobol' $(t,s)$-sequence $(m \ge t \ge 0)$  on $[0, 1]^{d+1}$
is used. 
Then we have
\begin{equation*}
    \sup_{\x,\x' \in \mc{X}} \mathbb{E} \left[ |K_M(\mathbf{x}, \mathbf{x}') - K(\mathbf{x}, \mathbf{x}')|^2 \right] 
  \leq 
  C^2 \cdot \frac{2^{2t+7(d+1)}}{{d}!}  \frac{(\log_2 M)^d}{M^2},
\end{equation*}
where the constant $C = 4\pi  D^{|u\backslash\{d+1\}|} \prod_{i\in u\backslash\{d+1\}} C_i$ as specified in Remark~\ref{rk:cond2}, if $K(\cdot, \cdot)$ is a shift-invariant kernel satisfying Condition \ref{cond1}; 
and if $K(\cdot, \cdot)$ is a general kernel with a square integrable integrand 
    satisfying  Condition \ref{assumption:owen_boundary_growth_condition} with all $A_j=0$, the constant $C$ is the same as that specified in Condition \ref{assumption:owen_boundary_growth_condition}.
\end{theorem}

\begin{remark}
    Compared with the deterministic error bound in \citet[Theorem 2.2]{huangquasi} for 
    $ 
    \sup_{\x,\x' \in \mc{X}}
    \left[ |K_M(\mathbf{x}, \mathbf{x}') - K(\mathbf{x}, \mathbf{x}')|^2 \right] 
    $, the average-case error bound in Theorem \ref{thm:sup_meansquare} reduces the exponent of $\log M$ from $4d+2$ to $d$ (where $M$ denotes the number of random features). This improved bound aligns with its better empirical performance in higher dimensions, as observed in practice (see Section~\ref{sec:simulation}).
\end{remark}

\subsection{Deterministic Error Bound}
\label{sec:deterministic-case}
In this subsection, we establish deterministic error bounds for the RQMC-based kernel approximation and integral operator approximation.

The following theorem (proved in Appendix~\ref{pf:thm:qmc1_deterministiccase_kernelapprox_error}) provides a deterministic kernel approximation error bound for kernels satisfying Condition~\ref{cond1}.
\begin{theorem}
\label{thm:qmc1_deterministiccase_kernelapprox_error}
    Suppose $K(\cdot, \cdot)$ satisfies Condition \ref{cond1}, and an 
    scrambled $(t,m,d+1)$-net in base $b$
    on $[0,1]^{d+1}$
    with $m\ge t$
    satisfying $\mathcal{D}^{*}\left(\left\{\mathbf{h}_{i}\right\}_{i=1}^{M}\right) \leq$ $C \frac{\log ^{d} M}{M}$ for all $M=b^m$ is used. Then there exists a constant $C^{\prime}>0$ (depending on $C, \mathcal{X} \subset \mathbb{R}^{d}$ and $K$) such that for all $M=b^m$,
$$
\sup _{\mathbf{x}, \mathbf{x}^{\prime}}\left|K_{M}\left(\mathbf{x}, \mathbf{x}^{\prime}\right)-K\left(\mathbf{x}, \mathbf{x}^{\prime}\right)\right|
\leq C^{\prime} \frac{ \log^{2d} M }{M}.
$$
\end{theorem}
\begin{remark}
    Compared with the upper bound for the Halton sequence~\citep[Theorem 2.2]{huangquasi}, the power of $\log M$ in Theorem~\ref{thm:qmc1_deterministiccase_kernelapprox_error} is reduced by $1$. This is achieved by requiring $M$ as a power of $b$. Note that the proof technique of Theorem~\ref{thm:qmc1_deterministiccase_kernelapprox_error} can be applied to QMC methods as well (i.e., QMC features using digital nets), and thus may be seen as an extension of~\citet[Theorem 2.2]{huangquasi}.
\end{remark}

\begin{remark}
    As one may expect, compared with the average-case error bound in Theorem~\ref{thm:sup_meansquare} for
    $
    \sup_{\x,\x' \in \mc{X}} \E 
    \left[ |K_M(\mathbf{x}, \mathbf{x}') - K(\mathbf{x}, \mathbf{x}')|^2 \right] 
    $,
     the deterministic error bound in Theorem~\ref{thm:qmc1_deterministiccase_kernelapprox_error} has a larger exponent of $\log M$.
\end{remark}

The above determinstic bound is very useful for establishing other kernel-related estimation bounds, e.g., for approximating the integral operator as shown in the proposition below.

For kernel ridge regression, suppose $(\mathbf{X}, Y) \in \mathcal{X} \times \mathbb{R}$ follows a distribution $P_{\mathbf{X} Y}$ with marginal distributions $P_{\mathbf{X}}$ and $P_{Y}$. Given the kernel function $K$, the integral operator $L: L^{2}\left(P_{\mathbf{X}}\right) \rightarrow L^{2}\left(P_{\mathbf{X}}\right)$ is defined as:
\begin{equation}
\label{eq:integral-operator}
L f(\mathbf{x}):=\mathbb{E}_{\mathbf{X} \sim P_{\mathbf{X}}}[K(\mathbf{X}, \mathbf{x}) f(\mathbf{X})].
\end{equation}
Define its approximation $L_{M}: L^{2}\left(P_{\mathbf{X}}\right) \rightarrow L^{2}\left(P_{\mathbf{X}}\right)$ as
\begin{equation*}
    L_{M} f(\mathbf{x}):=\mathbb{E}_{\mathbf{X} \sim P_{\mathbf{X}}}\left[K_{M}(\mathbf{X}, \mathbf{x}) f(\mathbf{X})\right].
\end{equation*}

The following proposition on the approximation error of the integral operator can be shown, using the same technique as in \citet[Proposition 2.6]{huangquasi}, which will be used in the proof of the theoretical properties of RQMC features in kernel ridge regression.
\begin{proposition}
    Under the same conditions as in Theorem~\ref{thm:qmc1_deterministiccase_kernelapprox_error}, we have
    $$
\left\|L_{M}-L\right\| \leq C^{\prime}  \frac{\log^{2d} M  }{M},
$$
where $\|\cdot\|$ denotes the operator norm.
\end{proposition}

For general kernels, we can also establish the deterministic error bound for the RQMC-based kernel approximation, if the following condition \citep{huangquasi} holds.
\begin{condition}
\label{cond:bv}
    Suppose there exists a function $\psi:\mathcal{X}\times [0,1]^p\to\R$ such that
$$K(\x,\x')=\int_{[0,1]^p}\psi(\x,\omega)\psi(\x',\omega) {\rm d} \omega,$$
and for any $\x,\x'\in\mathcal{X}$, $g(\omega)=\psi(\x,\omega)\psi(\x',\omega)$ is of bounded Hardy-Krause variation $V_{\rm HK}(g)\leq C_0$, for some $C_0>0$.
\end{condition}

It was shown in \citet{huangquasi} that 
the min kernel, Brownian bridge kernel, a class of iterative kernel, natural cubic spline kernel, and a class of product kernels satisfy Condition~\ref{cond:bv}. 
Assuming Condition~\ref{cond:bv} holds, the following theorem is a consequence of the Koksma-Hlawka inequality.
\begin{theorem}
\label{thm:QMC2}
Suppose $K(\cdot,\cdot)$ satisfies Condition~\ref{cond:bv}. 
Suppose an RQMC sequence on $[0,1]^{d+1}$ satisfying $\mathcal{D}^{*}\left(\left\{\mathbf{h}_{i}\right\}_{i=1}^{M}\right) \leq C \frac{\log ^{a} M}{M}\, (a>0)$   is used.
For any $\x,\x'\in \mathcal{X}$,
we have
$$|K_M(\x,\x')-K(\x,\x')|\leq C_0
C
\cdot \frac{ \log^a M }{M},$$
where $C_0$ is the constant in the Condition~\ref{cond:bv}.
\end{theorem}
\begin{remark}
Plenty of RQMC sequences satisfy the low discrepancy property required by Theorem~\ref{thm:QMC2}.
For example, such sequences can be obtained by applying Owen's scrambling~\citep{owen1995randomly} to a $(t,m,d+1)$-net (or a 
$(t,d+1)$-sequence), or by applying Cranley-Patterson rotation~\citep{cranley1976randomization} to a low-discrepancy sequence, such as the Halton sequence, a $(t,m,d+1)$-net or a 
$(t,d+1)$-sequence.
\end{remark}

\section{Application in Kernel Ridge Regression}
\label{sec:KRR}

In this section, we study the application of RQMC features in kernel ridge regression (KRR) and discuss how the use of RQMC features, as introduced in Section~\ref{sec:approx_kernel_func_with_rqmc}, can improve the computational performance over standard MC random features, without loss of theoretical error rate.

We start with an integrated overview of KRR and its computational approximation through both MC random features and RQMC features. 
Then theoretical guarantees are given for our proposed RQMCF-KRR method.
The presented formulation and results draw upon established literature on kernel methods and KRR with MC and QMC features\citep{huangquasi,scholkopf2002learning,caponnetto2007optimal,smale2007learning,bach2017equivalence,rudi2017generalization,avron2017random,rahimi2007random}.

\subsection{Background on Kernel Ridge Regression}
Consider a supervised learning setup with $n$ i.i.d.\ samples $(\x_i,y_i)_{i=1}^n$, where $\x_i \in \mathcal{X}$ and $y_i \in \mathbb{R}$, drawn from a distribution $P_{\X Y}$. The target function is the conditional expectation $f_{\ast}(\x)=\mathbb{E}[Y|\X=\x]$. Let $K:\mathcal{X}\times\mathcal{X}\to\mathbb{R}$ be a positive definite kernel associated with a reproducing kernel Hilbert space (RKHS) $\mathcal{H}$. The  KRR estimator \citep{scholkopf2002learning,caponnetto2007optimal} solves:
\begin{equation}\label{eq:KRR_obj_main}
    \hat{f}_\lambda := \arg\min_{f\in\mathcal{H}}\left\{\frac{1}{n}\sum_{i=1}^n (y_i - f(\x_i))^2 + \lambda \|f\|_{\mathcal{H}}^2\right\},
\end{equation}
with a regularization parameter $\lambda > 0$. The closed-form solution for \eqref{eq:KRR_obj_main} is:
\begin{equation}\label{eq:KRR_sol_main}
    \hat{f}_\lambda(\x) = \sum_{i=1}^n \hat{\alpha}_i K(\x_i,\x), \quad \text{where } \hat{\boldsymbol{\alpha}} = (\mathbf{K}+n\lambda \mathbf{I}_n)^{-1}\mathbf{y}
\end{equation}
and $\mathbf{K}=[K(\x_i,\x_j)]_{i,j=1}^n$. Although KRR achieves minimax-optimal rates~\citep{caponnetto2007optimal},
its direct implementation costs $O(n^3)$ in time and $O(n^2)$ in memory.

\subsection{Monte Carlo Random Feature Approximations}
A popular strategy to scale KRR is to approximate the kernel $K$ using random features. Suppose the kernel $K$ admits an integral representation:
\begin{equation}\label{eq:int_rep_rf}
    K(\x,\x')=\int_{\Omega}\psi(\x,\omega)\psi(\x',\omega)\,\mathrm{d}\pi(\omega).
\end{equation}
With $M$ independent samples $\{\omega_i\}_{i=1}^M$ from $\pi$, we have the MC approximation:
\begin{equation}\label{eq:K_approx_mc_main}
    K_M(\x,\x') = \frac{1}{M}\sum_{i=1}^M \psi(\x,\omega_i)\psi(\x',\omega_i).
\end{equation}
By setting $\pph_M(\x)=\frac{1}{\sqrt{M}}(\psi(\x,\omega_1),\ldots,\psi(\x,\omega_M))^\top$, $K_M(\x,\x')$ can be written as $K_M(\x,\x')=\pph_M(\x)^\top\pph_M(\x')$. Replacing $K$ with $K_M$ in \eqref{eq:KRR_sol_main} leads to a random feature-based KRR (RF-KRR) estimator \citep{rudi2017generalization,avron2017random}. This method reduces the complexity to $O(nM^2+M^3)$ in time and $O(nM)$ in memory. 
Moreover, it is known that 
if $M \asymp n^{\frac{2r}{2r+1}}$ (up to logarithmic factors), RF-KRR preserves the same statistical guarantees as KRR, where $r\in[\frac{1}{2},1]$ is a smoothness parameter of the underlying true regression function~\citep{rudi2017generalization,huangquasi}.

\subsection{Randomized Quasi-Monte Carlo Features and Improved Approximations}
While MC random features yield a typical convergence rate of $O_P(M^{-1/2})$ for the kernel approximation error, QMC and RQMC methods can often achieve better rates $O(M^{-1})$ (up to logarithmic factors) by using low-discrepancy sequences instead of i.i.d.\ random samples.
\citet{huangquasi} proposed QMCF-KRR, which uses Quasi-Monte Carlo sequence, and in particular, Halton sequence, for the kernel approximation in KRR. This method works well in the low dimensional settings. However, it was found that when the dimension is larger than 10, the performance of QMCF-KRR degrades and may even be worse than RF-KRR~\citep{huangquasi}.

We propose RQMC-feature-based KRR (RQMCF-KRR), with scrambled net used for the kernel approximation in the KRR.
Substituting scrambled net sequences $\{\tilde{\omega}_i\}_{i=1}^M$ into \eqref{eq:int_rep_rf}, we approximate the kernel $K$ by 
\begin{equation}
\label{eq:K_approx_rqmc}
    K_M(\x,\x') := \frac{1}{M}\sum_{i=1}^M \psi(\x,\tilde{\omega}_i)\psi(\x',\tilde{\omega}_i),
\end{equation}
thus defining the RQMCF-KRR. As will be seen in Section~\ref{sec:theory-results-rqmcf-krr}, RQMCF-KRR requires fewer features $M$ than RF-KRR to attain the same statistical precision. Specifically, to obtain the optimal rates, RQMCF-KRR requires $M$ only of order $n^{\frac{1}{2r+1}}$, where $r\in[\frac{1}{2},1]$ characterizes the smoothness of the true regression function. This is an improvement over the $M \asymp n^{\frac{2r}{2r+1}}$ required by RF-KRR, resulting in a more efficient computational trade-off without compromising statistical performance. Theoretically, RQMCF-KRR achieves the same computational complexity as QMCF-KRR. In practice, RQMCF-KRR appears to be more suitable for higher-dimensional problems, as will be illustrated by the empirical performance in Section~\ref{sec:simulation}.

\subsection{Theoretical Results for RQMCF-KRR}\label{sec:theory-results-rqmcf-krr}
We adopt the same KRR conditions as in \citet{huangquasi}:
\begin{condition2}
    \label{cond:kernel}
   (i) $K(\x,\x')$ is continuous and has the integral representation \eqref{eq:int_rep_rf},
in which $|\psi(\x,\omega)|\leq \kappa$ for some constant $\kappa>0$. Assume $\X$ has full support on $\mathcal{X}$, and $\omega\mapsto \psi(\cdot,\omega)$, as a map from $\Omega$ to $L^2(P_\X)$, is continuous.

\vspace{-0.1cm}
(ii) $\pi$ in \eqref{eq:int_rep_rf} is the uniform distribution over $[0,1]^p$ for some $p\geq 1$, and 
an RQMC sequence is used for approximating the kernel as in 
\eqref{eq:K_approx_rqmc}, from which we have
\vspace{-0.25cm}
$$\sup_{\x,\x'\in \cX} 
\left| K(\x,\x')- K_M(\x,\x')\right| \leq C\cdot  \frac{\log^a M}{M}$$ 
for some positive constants $C$ and $a$. %
\end{condition2}
\begin{condition2}
    \label{cond:y_tail}
The distribution of $Y$ satisfies a Bernstein condition:
there exist positive constants $\sigma$ and $D$ such that $\mathbb{E}[|Y|^k \mid \X]\leq \frac{1}{2}k!\sigma^2 D^{k-2}$ for all $k\geq 2$. 
\end{condition2}
\begin{condition2}
    \label{cond:f=Lrg}
There exists $r\in [1/2,1]$ such that $f_\mathcal{H} = L^r g$ for some $g\in L^2(P_\X)$, where $f_\mathcal{H}$ solves $\min_{f\in\mathcal{H}} \mathcal{E}(f)$, and $L$ is the integral operator defined in \eqref{eq:integral-operator}. Let $R := \max\{\|g\|_{L^2(P_\X)},1\} $ be a positive constant.
\end{condition2}
\begin{remark}
The above conditions hold under mild conditions. 
Theorem~\ref{thm:qmc1_deterministiccase_kernelapprox_error} guarantees that when a shift-invariant kernel satisfying Condition~\ref{cond1} and a scrambled net are used,
    KRR Condition~\ref{cond:kernel}(ii) holds.
    In addition, KRR Condition~\ref{cond:kernel}(ii) also holds for a general kernel satisfying Condition~\ref{cond:bv}, by Theorem~\ref{thm:QMC2}.
    KRR Condition~\ref{cond:y_tail} is a usual tail condition on the response variable, which holds for the sub-exponential distribution.
    KRR Condition~\ref{cond:f=Lrg} can be viewed as a smoothness condition on the true regression function and is widely adopted in the kernel machine literature~\citep{smale2003estimating,caponnetto2007optimal}. See \citet{huangquasi} for more detailed discussions on these conditions.
\end{remark}

Theorem \ref{Thm:RQMCF-KRR} below (see Appendix~\ref{pf:thm:rqmcf-krr} for a proof) establishes the statistical error rate of the proposed RQMCF-KRR estimator. 
\begin{theorem}
\label{Thm:RQMCF-KRR}
Assume KRR Conditions~\ref{cond:kernel}, \ref{cond:y_tail}, \ref{cond:f=Lrg}. 
Let $\lambda = \tilde{C}n^{-\frac{1}{2r+1}}\in(0,e^{-1}]$,
 and $\hat f_{\lambda,M} $ be defined as in \eqref{eq:KRR_sol_main}.
Then 
$M=\frac{\log^a(1/\lambda)}{\lambda}=n^{\frac{1}{2r+1}}\log^a(n^{\frac{1}{2r+1}}/\tilde{C})/\tilde{C}$
is enough to guarantee that,
for any $\delta \in (0,1]$, there exists $n_0$ (of order $(\log \frac{1}{\delta})^{1+\frac{1}{2r}} $), such that when $n\geq n_0$, with probability at least $1-\delta$, the excess risk
\begin{equation}
    \label{eq:risk2}\mathcal{E}(\hat{f}_{\lambda,M} ) - \inf_{f\in\mathcal{H}}\mathcal{E}(f) \leq 
    C_1 n^{-\frac{2r}{2r+1}} \log^2 \frac{6}{\delta},
\end{equation}
where $C_1$ is a constant depending only on $ \kappa,\sigma, D, R, r,\tilde{C},$ $C$ and $a$.
\end{theorem}
    The error bound in \eqref{eq:risk2} matches the statistical convergence rate established for exact kernel ridge regression (KRR) \citep[Theorem 1]{caponnetto2007optimal} and for random features-based KRR (RF-KRR) \citep[Theorem 2]{rudi2017generalization}.

    Our RQMCF-KRR approach is more computationally efficient under smoother conditions.
To illustrate this, consider that RF-KRR, as shown in \citep[Theorem 2]{rudi2017generalization}, requires  the order of 
\(
   M \asymp n^{\frac{2r}{2r+1}} \log\left(\frac{108 \kappa^2 n}{\delta}\right)
\)
random features to achieve an excess risk of the order of 
\(
   \tildeC_1 n^{-\frac{2r}{2r+1}} \log^2\left(\frac{18}{\delta}\right).
\)
In contrast, our RQMCF-KRR method requires only 
\(
   M = n^{\frac{1}{2r+1}}\log^a\left(\frac{n^{\frac{1}{2r+1}}}{\tildeC}\right)/\tildeC
\)
features to attain the same statistical accuracy, where $r \in [1/2,1]$.
For $r > 1/2$, RQMCF-KRR enables a substantial reduction in the required number of features. Ignoring constant and logarithmic factors, RQMCF-KRR requires only the order of $n^{\frac{1}{2r+1}}$ features, which is strictly smaller than $n^{\frac{2r}{2r+1}}$ required by RF-KRR, thereby reducing the computational cost significantly.

In addition, note that the RQMCF-KRR achieves the same computational complexity as the QMCF-KRR proposed in \citet{huangquasi}, while exhibiting superior performance in higher dimensions than QMCF-KRR, as shown in Section~\ref{sec:simulation} below.

\section{Simulations}
\label{sec:simulation}
In this section, 
we show the superior performance of RQMC methods in kernel approximation and kernel ridge regression. In particular, we present simulation results on kernel approximation for the average case and deterministic  case discussed in Section~\ref{sec:approx_kernel_func_with_rqmc}, as well as the simulations for the KRR results in Section~\ref{sec:KRR}.
For QMC features, we follow \citet{huangquasi}'s proposal.
For RQMC features, 
we use the scrambled Sobol' sequence implemented in the Python SciPy package.

\subsection{Simulations on Kernel Approximation}

\paragraph{Average Case}
Theorems~\ref{thm:kernel-variance} and \ref{thm:sup_meansquare} provide theoretical guarantees on the average-case approximation accuracy of the RQMC features. Here, we examine the performance in practice.
We consider Gaussian kernel 
\(
K(\mathbf{x},\mathbf{x}') = \exp \Bigl(-\tfrac{1}{2\sigma^2}\|\mathbf{x}-\mathbf{x}'\|^2 \Bigr) \).
Let \(\mathbf{X}, \mathbf{X}' \) be i.i.d.\ from \(\text{Unif}[0,1]^d\).
The bandwidth \(\sigma\) of the Gaussian kernel is chosen to be the median of \(\|\mathbf{X} - \mathbf{X}'\|\) (computed numerically).
We sample \(10^3\) \((\mathbf{x}, \mathbf{x}')\) pairs, where \(\mathbf{x}\) and \(\mathbf{x}'\) are i.i.d.\ drawn from \(\text{Unif}[0,1]^d\). For each pair \((\mathbf{x}, \mathbf{x}')\), the same set of QMC features is used to compute \( \bigl\lvert K(\mathbf{x}, \mathbf{x}') - K_M(\mathbf{x}, \mathbf{x}')\bigr\rvert^2 \), given the deterministic nature of QMC points. In contrast, \(10^3\) independent sets of MC and RQMC features are sampled to compute the average square error (a numerical estimate of \( \mathbb{E}_\omega\bigl[\bigl\lvert K(\mathbf{x}, \mathbf{x}') - K_M(\mathbf{x}, \mathbf{x}')\bigr\rvert^2\bigr] \)). We then take both the supremum and the average of these errors over the \(10^3\) sampled pairs.

In Figure~\ref{fig:E_x-E_w}, we plot the average error over the \(10^3\) \((\mathbf{x}, \mathbf{x}')\) pairs as a function of the number of random features in various dimensions. In low-dimensional settings, RQMC features perform similarly to QMC features, and both outperform MC features. However, as the dimension increases, QMC features degrade, whereas RQMC features continue to perform comparably to or better than MC features.

In Figure~\ref{fig:sup_x-E_w}, we illustrate the supremum error over the same \(10^3\) \((\mathbf{x}, \mathbf{x}')\) pairs across dimensions. Even in moderately low-dimensional cases (e.g., when the dimension is greater than 1), QMC features do not achieve a high-accuracy kernel approximation. In contrast, RQMC features exhibit better performance than MC features in lower dimensions, and their performances become increasingly similar as the dimension grows.

\paragraph{Deterministic Case}
Theorem~\ref{thm:qmc1_deterministiccase_kernelapprox_error} provides desirable theoretical guarantee for the deterministic approximation error bound of the RQMC features, and we examine its empirical performance here.
The same Gaussian kernel as above is considered. 
We sample \(10^4\) \((\mathbf{x}, \mathbf{x}')\) pairs, with \(\mathbf{x}\) and \(\mathbf{x}'\) drawn i.i.d.\ from \(\text{Unif}[0,1]^d\). For each pair, one set of MC, QMC, and RQMC features is generated to compute \( \bigl\lvert K(\mathbf{x}, \mathbf{x}') - K_M(\mathbf{x}, \mathbf{x}')\bigr\rvert^2 \). We take the supremum of these errors over the \(10^4\) pairs to numerically estimate \( \sup_{\mathbf{x}, \mathbf{x}' \in \mathcal{X}} \bigl\lvert K(\mathbf{x}, \mathbf{x}') - K_M(\mathbf{x}, \mathbf{x}')\bigr\rvert^2 \).

Figure~\ref{fig:deterministic} shows the resulting supremum error for different dimensions. In low-dimensional cases, RQMC features perform on par with QMC features, and both outperform MC features. As the dimension increases, the performance of the QMC approach deteriorates, while RQMC features remain comparable to MC features.

\clearpage
\begin{figure}[htbp]
\thispagestyle{empty}
  \centering
  \includegraphics[width=\textwidth, height=\textheight, keepaspectratio]{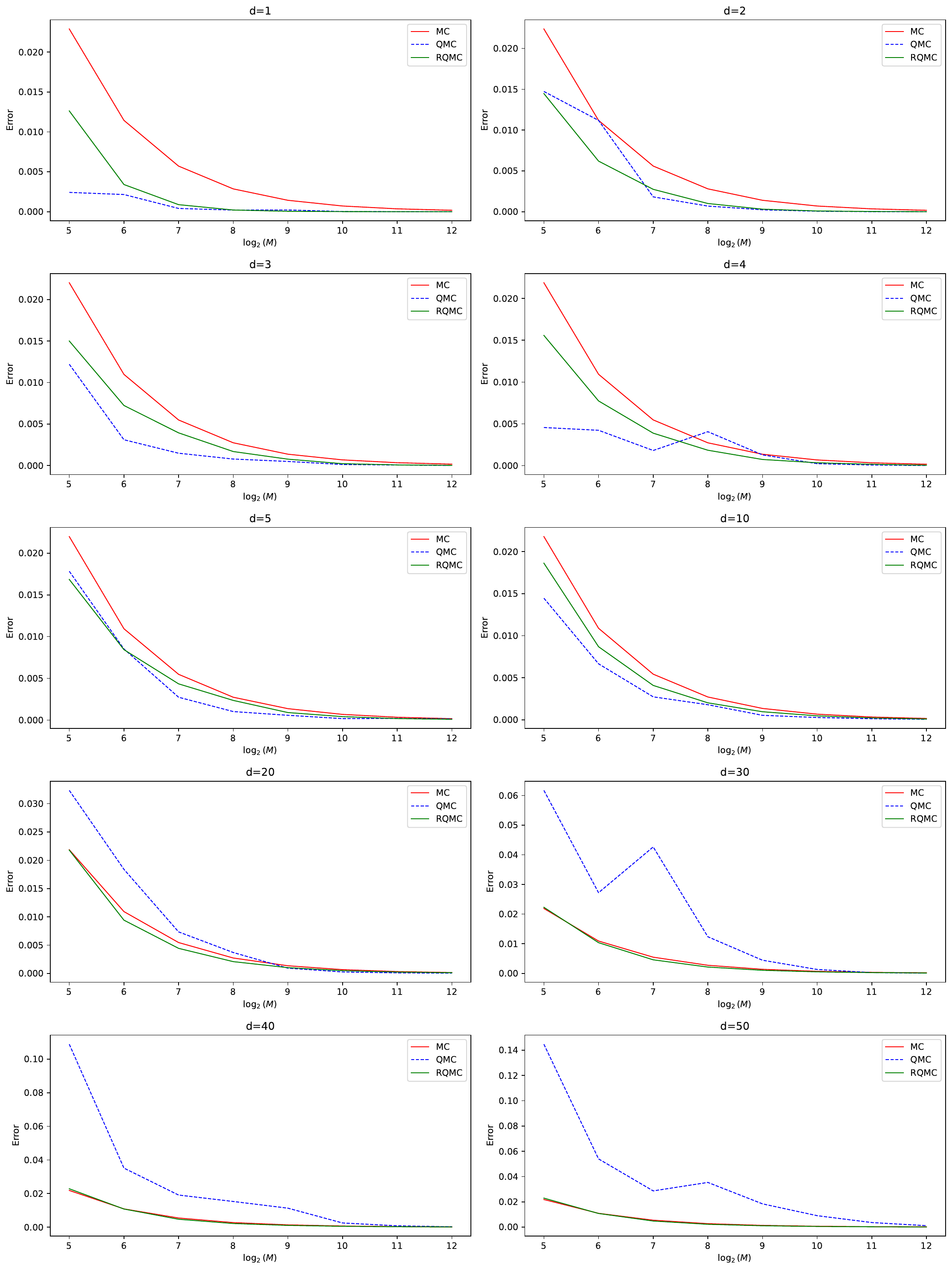}
  \caption{The average error $\E_{\x,\x' \in \mc{X}} \E_\omega \abs{K_M(\x,\x')-K(\x,\x')}^2$ against the number of random features for MC, QMC, RQMC based methods.}
  \label{fig:E_x-E_w}
\end{figure}
\clearpage

\clearpage
\begin{figure}[htbp]
\thispagestyle{empty}
  \centering
  \includegraphics[width=\textwidth, height=\textheight, keepaspectratio]{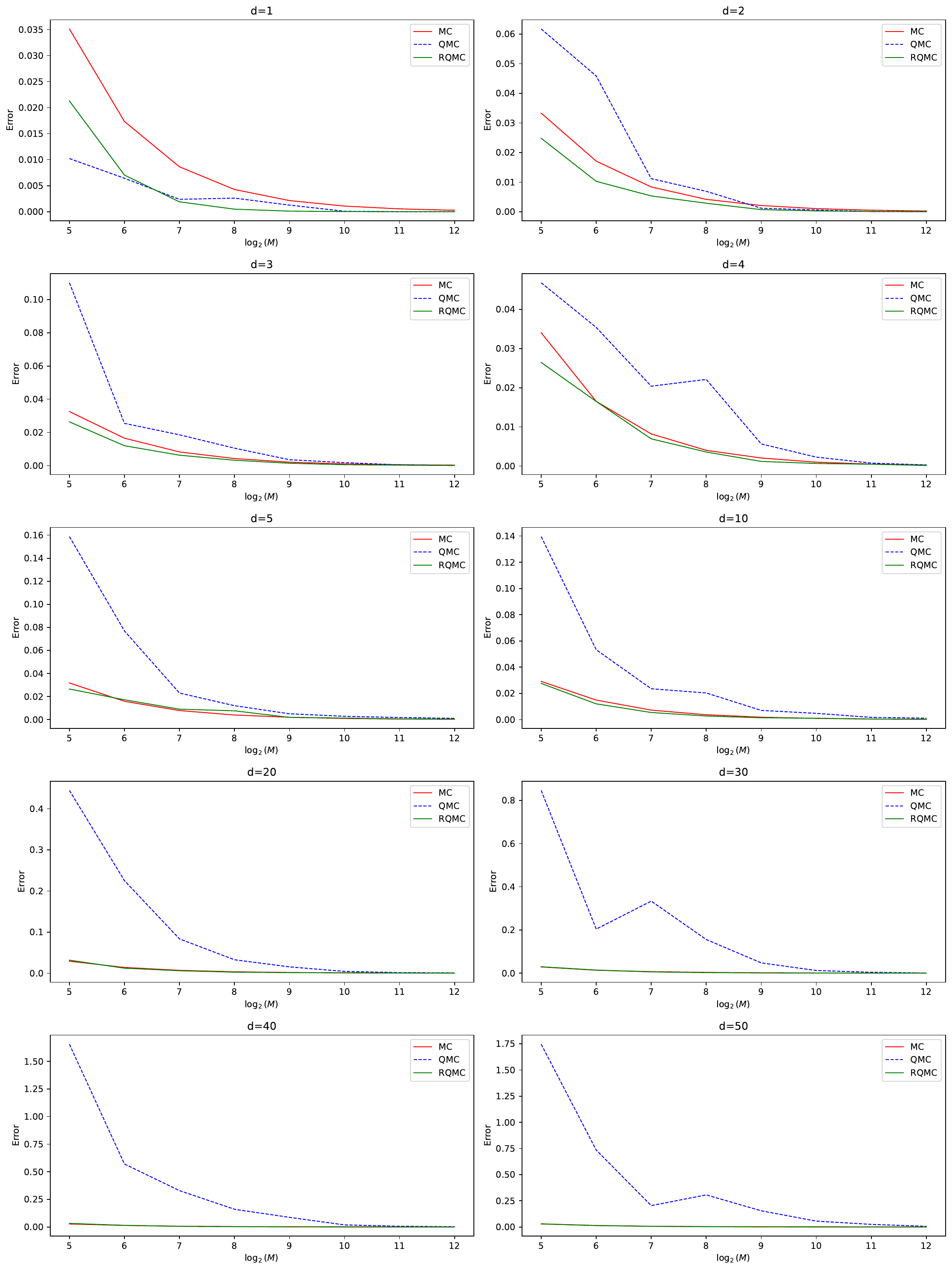}
  \caption{The sup-average error $\sup_{\x,\x' \in \mc{X}} \E_\omega \abs{K_M(\x,\x')-K(\x,\x')}^2$ against the number of random features for MC, QMC, RQMC based methods.}
  \label{fig:sup_x-E_w}
\end{figure}
\clearpage

\clearpage
\begin{figure}[htbp]
\thispagestyle{empty}
  \centering
  \includegraphics[width=\textwidth, height=\textheight, keepaspectratio]{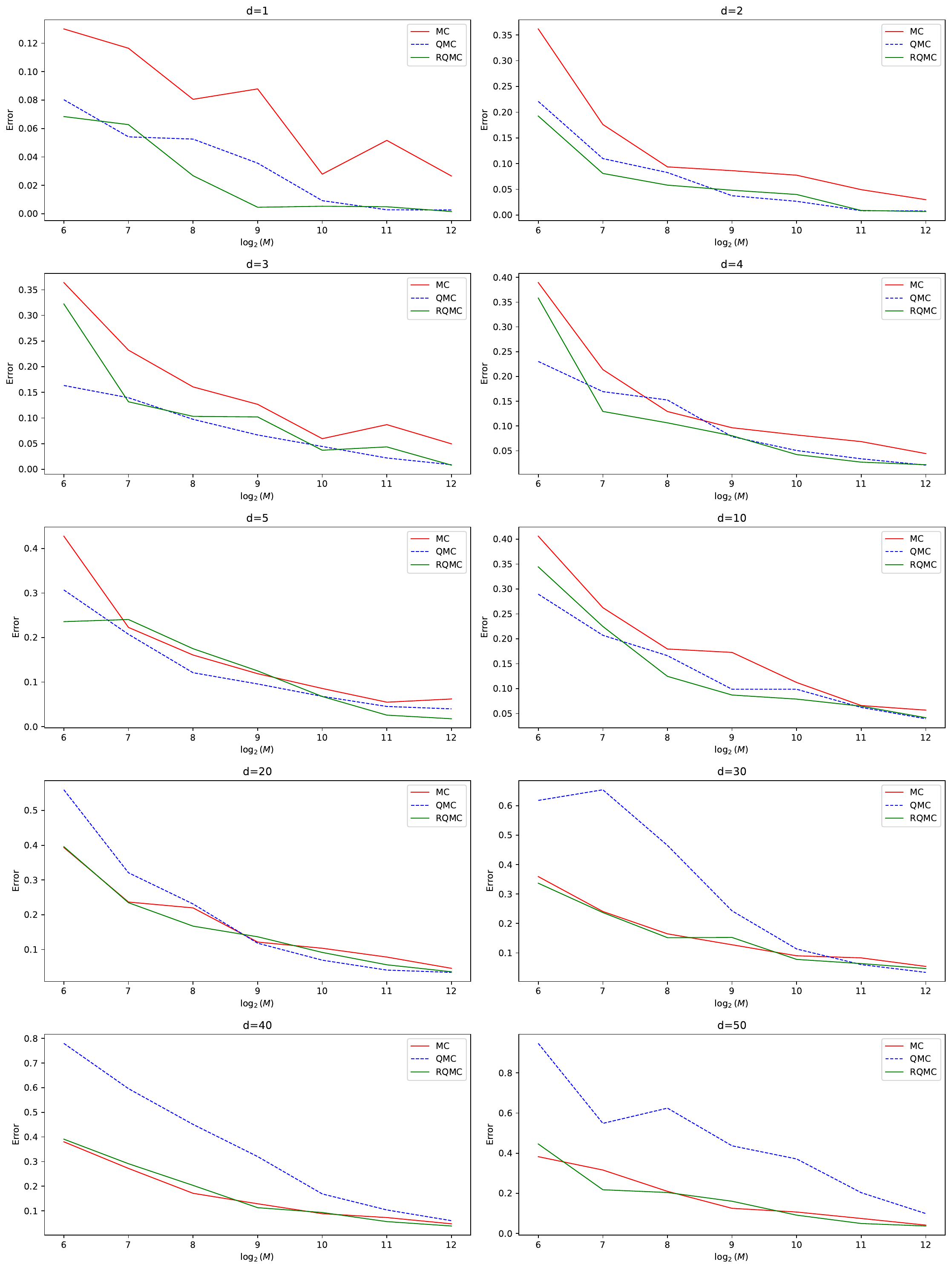}
  \caption{The determinstic error $\sup_{\x,\x' \in \mc{X}}  \abs{K_M(\x,\x')-K(\x,\x')}^2$ against the number of random features for MC, QMC, RQMC based methods.}
  \label{fig:deterministic}
\end{figure}
\clearpage

\subsection{Simulations on Kernel Ridge Regression}
\label{sec:simu-krr}
In this subsection, we compare the performance of RF-KRR, QMCF-KRR, RQMCF-KRR in low and moderately high dimensions
when modeling data with the Gaussian kernel. As shown in Theorem~\ref{Thm:RQMCF-KRR}, the RQMC method outperforms the MC method in smoother cases, specifically when $r \in [1/2,1]$ is large. 
We illustrate this with experiments for the scenario $r=1$. In fact, empirical evidence suggests that RQMCF-KRR also exhibits advantages when $r = 1/2$; additional simulations are provided in Appendix~\ref{sec:additional-simu}.

\paragraph{Experimental Setup}
We follow the experimental setting in~\citet{huangquasi} for simulations on kernel ridge regression.
We generate training and test data according to the model
$Y = f(\mathbf{X}) + \varepsilon$,
where \(\mathbf{X} \sim \text{Unif}[0,1]^d\) and \(\varepsilon \sim \mathcal{N}(0,1)\). We use the Gaussian kernel 
\(
K(\mathbf{x},\mathbf{x}') = \exp \Bigl(-\tfrac{1}{2\sigma^2}\|\mathbf{x}-\mathbf{x}'\|^2 \Bigr) \),
with bandwidth \(\sigma\) chosen to be the median of \(\|\mathbf{X} - \mathbf{X}'\|\) (computed numerically), where \(\mathbf{X}, \mathbf{X}' \) are i.i.d.\ from \(\text{Unif}[0,1]^d\).

For \(r=1\), any function \(\tilde{f}\) in \(\operatorname{ran} L^r\) can be written as
\(
\tilde{f}(\mathbf{x}) = \int K(\mathbf{x},\mathbf{z})\, g(\mathbf{z}) \, \mathrm{d}P_{\mathbf{X}}(\mathbf{z})
\)
for some \(g \in L^2(P_{\mathbf{X}})\). The function \(g(\mathbf{z}) = \exp\bigl(\tfrac{1}{2\sigma^2}\|\mathbf{z}\|^2\bigr)\) is adopted, which leads to a closed form
\(
\tilde{f}(\mathbf{x}) 
= \sigma^{2d} \exp\Bigl(-\tfrac{1}{2\sigma^2}\|\mathbf{x}\|^2\Bigr) \prod_{j=1}^d \frac{\exp(\tfrac{x_j}{\sigma^2}) - 1}{x_j}
\).
To control the signal-noise-ratio, we set
\(
f(\mathbf{x}) = C_{\tilde{f}} \cdot \tilde{f}(\mathbf{x}),
\)
where \(C_{\tilde{f}}\) is chosen such that the mean of \(f(\mathbf{X})\) equals 5. The kernel ridge regularization parameter is fixed as 
\(
\lambda = 0.25 \, n^{-\tfrac{1}{2r+1}}
\).

\paragraph{Results}
Figure~\ref{fig:r-1} shows the test mean square error (MSE) against the number of random features for the exact KRR, RF-KRR, QMCF-KRR and RQMCF-KRR under different values of dimension \(d\). Specifically, we generate and hold fixed \(10^6\) test data points,
  and consider 1000 realizations of training samples, each of size \(10^4\). For each realization, we train with different methods and record their test errors (MSE) on the fixed test set.
The solid lines in Figure~\ref{fig:r-1} show the average test MSE over 1000 trials, and the shaded areas indicate the 25\% and 75\% quantiles. 

It can be observed that RQMCF-KRR outperforms RF-KRR in both low and higher dimensions. In the low dimensional setting, RQMCF-KRR 
substantially reduces the number of features needed to attain a comparable generalization error to that of the exact KRR, relative to the MC-based random features. As the dimension increases, their performances get closer, but RQMC features still exhibit superior or similar performance. 

In low-dimensional settings, RQMCF-KRR and QMCF-KRR exhibit comparable performance. However, as the dimension increases, QMCF-KRR experiences a substantial decline in effectiveness, whereas RQMCF-KRR remains stable.

Results for the case \(r=0.5\) yield similar conclusions; interested readers are referred to Appendix~\ref{sec:additional-simu} for more details.

\clearpage
\begin{figure}[htbp]
\thispagestyle{empty}
  \centering
  \includegraphics[width=\textwidth, height=\textheight, keepaspectratio]{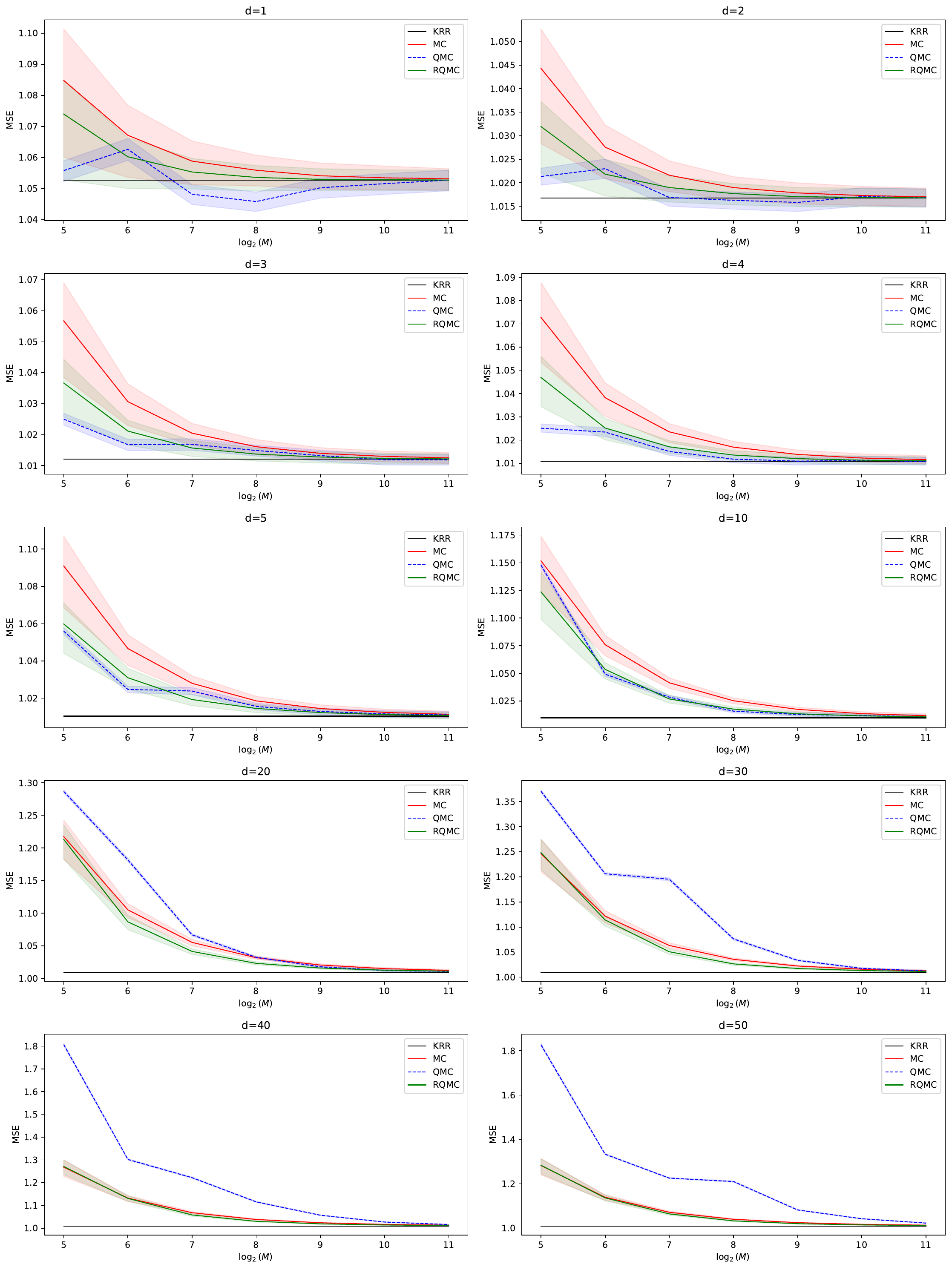}
  \caption{The test MSE against the number of random features ($r=1$), for exact KRR, RF-KRR, QMCF-KRR and RQMCF-KRR.}
  \label{fig:r-1}
\end{figure}
\clearpage

\bibliographystyle{ims}
\bibliography{reference}

\newpage 
\appendix
\numberwithin{equation}{section}

\section{Proof of the Results in Section~\ref{sec:approx_kernel_func_with_rqmc} and \ref{sec:KRR}
}
\label{sec:appendix:pf}

\subsection{Proof of Theorem~\ref{thm:qmc1_averagecase_kernelapprox_error}}
\label{pf:thm:qmc1_averagecase_kernelapprox_error}

\begin{proof}
Consider $f:[0,1]^{d+1} \rightarrow \mathbb{R}$, which takes $(\mathbf{t}, b) \mapsto \sqrt{2} \cos \left(\mathbf{x}^{\top} \boldsymbol{\Phi}^{-1}(\mathbf{t})+2 \pi b\right)$. Let $\tilde{f}_{M}$ be the low variation function that coincides with $f$ on a ``large set" $K_{M}=\left[\varepsilon_{M}, 1-\varepsilon_{M}\right]^{d+1}$ as defined in \citet[Appendix B.1]{huangquasi}.
We have

\begin{equation*}
    \begin{split}
    \left|\int_{[0,1]^{d+1}} f(\mathbf{x}) \mathrm{d} \mathbf{x}-\frac{1}{M} \sum_{i=1}^{M} f\left(\mathbf{h}_{i}\right)\right| \leq \int_{[0,1]^{d+1}}
    &\left|f(\mathbf{x})-\tilde{f}_{M}(\mathbf{x})\right| \mathrm{d} \mathbf{x}\\
    & +\mathcal{D}^{*}\left(\left\{\mathbf{h}_{i}\right\}_{i=1}^{M}\right) V_{\mathrm{HK}}\left(\tilde{f}_{M}\right)\\
    & +\frac{1}{M} \sum_{i=1}^{M}\left|\tilde{f}_{M}\left(\mathbf{h}_{i}\right)-f_{M}\left(\mathbf{h}_{i}\right)\right|.
    \end{split}
\end{equation*}

When $\left\{\mathbf{h}_{i}\right\}_{i=1}^{M}$ is an RQMC sequence, each $\mathbf{h}_{i}$ marginally follows Unif $[0,1]^{d+1}$. Therefore, by taking expectation,

\begin{equation*}
    \begin{split}
\mathbb{E}\left|\int_{[0,1]^{d+1}} f(\mathbf{x}) \mathrm{d} \mathbf{x}-\frac{1}{M} \sum_{i=1}^{M} f\left(\mathbf{h}_{i}\right)\right| \leq 2 \int_{[0,1]^{d+1}}
& \left|f(\mathbf{x})-\tilde{f}_{M}(\mathbf{x})\right| \mathrm{d} \mathbf{x}\\
& +\mathcal{D}^{*}\left(\left\{\mathbf{h}_{i}\right\}_{i=1}^{M}\right) V_{\mathrm{HK}}\left(\tilde{f}_{M}\right).
\end{split}
\end{equation*}
By \citet[Inequality B.4]{huangquasi},
$$
V_{\mathrm{HK}}\left(\tilde{f}_{M}\right) \leq 2 B\left(1-2 \log 2-2 \log \varepsilon_{M}\right)^{d},
$$
where $B=4\pi D^{|u\setminus \{d+1\}|} \prod_{i \in u \setminus \{d+1\}}C_i$, 
with 
$$D=\max_{\x,\y \in \mc{X}, \, i \in \{1,\ldots,d\}} \{ \abs{x_i-y_i}, \abs{x_i+y_i}  \}.$$

By \citet[Inequality B.6]{huangquasi},
$$
\int_{[0,1]^{d+1}}\left|f(\mathbf{x})-\tilde{f}_{M}(\mathbf{x})\right| \mathrm{d} \mathbf{x} \leq 3 \cdot 2^{d-1} B \varepsilon_{M}\left(2+(2-\log 2) d-d \log \varepsilon_{M}\right).
$$
Therefore,\\
\begin{equation*}
    \begin{split}
\mathbb{E}\left|\int_{[0,1]^{d+1}} f(\mathbf{x}) \mathrm{d} \mathbf{x}-\frac{1}{M} \sum_{i=1}^{M} f\left(\mathbf{h}_{i}\right)\right| &\leq 3 \cdot 2^{d} B 
 \varepsilon_{M}\left(2
+(2-\log 2) d-d \log \varepsilon_{M} \right) \\
&+\frac{C \log ^{d+1} M}{M} 2 B\left(1-2 \log 2-2 \log \varepsilon_{M}\right)^{d}.
   \end{split}
\end{equation*}
By taking $\varepsilon_{M}=1 / M$, we have
$$
\mathbb{E}\left|\int_{[0,1]^{d+1}} f(\mathbf{x}) \mathrm{d} \mathbf{x}-\frac{1}{M} \sum_{i=1}^{M} f\left(\mathbf{h}_{i}\right)\right| \leq C^{\prime} \cdot \frac{\log ^{2 d+1} M}{M}.
$$
\end{proof}

\subsection{Proof of Theorem~\ref{thm:sup_meansquare}}
\label{pf:thm:sup_meansquare}

In this part, we prove Theorem~\ref{thm:sup_meansquare}, 
following the proof strategy proposed in~\citet{liu2024randomized,dick2010digital}.
We begin by introducing some notations.
Let $[d+1]$ denote the set $\{1,2,\ldots,d+1\}$.
For a square integrable function $f$ on $[0,1]^{d+1}$,
consider the Walsh decomposition of $f$~\citep{dick2010digital, liu2024randomized}:
\begin{equation*}
	f(\bm{t}) = \sum_{\bm{\ell} \in \mathbb{N}_0^{d+1}} \bar{f}({\bm{\ell}}) {}_{2}\wal_{\bm{\ell}}(\bm{t})
\end{equation*}
where $\bar{f}$ denotes the Walsh coefficients.
Fix $\bm{\ell} = \left({\ell}_1, \dotsc, {\ell}_{d+1} \right) \in \mathbb{N}_0^{d+1}$, 
and let
\begin{equation*}
	L_{\bm{\ell}} = \{ \bm{k} = ({k}_1, \dotsc, {k}_{d+1}) \in \mathbb{N}_0^{d+1}: \left\lfloor 2^{l_j - 1} \right\rfloor \leq k_j < 2^{l_j}\ \textrm{for}\ 1\leq j \leq {d+1} \}.
\end{equation*}
Then the Walsh expansion of $f$ corresponding to $L_{\bm{\ell}}$ is defined as 
$$	\beta_{\bm{\ell}}(\bm{t}) := \sum_{\bm{k}\in L_{\bm{\ell}}} \bar{f}(\bm{k}) {}_{2}\wal_{\bm{k}}(\bm{t}). 
$$
Define
\begin{equation}
	\label{eq:sigma2_l}
	\sigma^2_{\bm{\ell}}: =
    \sigma^2_{\bm{\ell}}(f): =
    \sum_{\bm{k}\in L_{\bm{\ell}}} \abs{\bar{f}(\bm{k})}^2.
\end{equation}
Let 
\begin{equation}
	\label{eq:set_T_ell}
	T_{\bm{\ell}} = \{ \bm{k} = ({k}_1, \dotsc, {k}_{d+1}) \in \mathbb{N}_0^{d+1}: 0 \leq k_j < 2^{l_j}\ \textrm{for}\ 1\leq j \leq {d+1} \}.
\end{equation}
Let $\Delta_j$ be the set difference, 
$
	\Delta_j T_{\bm{\ell}} := 
		 T_{\bm{\ell}} \setminus  T_{\bm{\ell} - \bm{e}_j},
$
where $\bm{e}_j$ is the standard basis vector. When $l_j = 0$, let $T_{\bm{\ell} - \bm{e}_j} \coloneqq \varnothing$. 
We further define the composition of set difference 
$$
\Delta_{j^\prime}(\Delta_j T_{\boldell}):= \Delta_{j^\prime}(T_{\bm{\ell}} \setminus  T_{\bm{\ell} - \bm{e}_j}):=(\Delta_{j^\prime} T_{\boldell})\setminus(\Delta_{j^\prime}T_{\bm{\ell} - \bm{e}_j})
$$
with $j^\prime \neq j$.
Then the set $L_{\bm{\ell}}$ can be expressed as the composition of set differences:
\begin{equation}
	L_{\bm{\ell}} = \left(\bigotimes_{j=1}^{{d+1}} \Delta_j \right) T_{\bm{\ell}}, 
\end{equation}
where $\bigotimes_{j=1}^{d+1} \Delta_j := \Delta_{d+1} \circ \Delta_{d} \circ \cdots \circ \Delta_1$.

Further define
	\begin{equation*}
        \mc{D}_j
        \sum_{\bm{k} \in T_{\bm{\ell}}} \bar{f}(\bm{k}) {}_{2}\wal_{\bm{k}}(\bm{t}) =
		\begin{cases} \sum_{\bm{k} \in T_{\bm{\ell}}} \bar{f}(\bm{k}) {}_{2}\wal_{\bm{k}}(\bm{t}) - \sum_{\bm{k} \in T_{\bm{\ell} - \bm{e}_j}} \bar{f}(\bm{k}) {}_{2}\wal_{\bm{k}}(\bm{t}) & \ \mathrm{if} \ l_j \geq 1,\\
		\sum_{\bm{k} \in T_{\bm{\ell}}} \bar{f}(\bm{k}) {}_{2}\wal_{\bm{k}}(\bm{t}) & \ \mathrm{if} \ l_j =0.
		\end{cases}
	\end{equation*}
And we define the composition $\mc{D}_{j^\prime} \mc{D}_j$ with $j^\prime \neq j$ similarly:
	\begin{equation*}
        \begin{split}
        \mc{D}_{j^\prime}
        \mc{D}_j
        &\sum_{\bm{k} \in T_{\bm{\ell}}} \bar{f}(\bm{k}) {}_{2}\wal_{\bm{k}}(\bm{t}) =\\
		&\begin{cases} 
        \mc{D}_{j^\prime}
        \sum_{\bm{k} \in T_{\bm{\ell}}} \bar{f}(\bm{k}) {}_{2}\wal_{\bm{k}}(\bm{t}) 
        - \mc{D}_{j^\prime}\sum_{\bm{k} \in T_{\bm{\ell} - \bm{e}_j}} \bar{f}(\bm{k}) {}_{2}\wal_{\bm{k}}(\bm{t}) & \ \mathrm{if} \ {\ell}_{j^\prime} \geq 1,\\
	\mc{D}_{j}	
        \sum_{\bm{k} \in T_{\bm{\ell}}} \bar{f}(\bm{k}) {}_{2}\wal_{\bm{k}}(\bm{t}) & \ \mathrm{if} \ {\ell}_{j^\prime} =0.
		\end{cases}
                \end{split}
	\end{equation*}

Then we have
	\begin{equation}
	\label{eq:walsh_series_L_ell}
		\begin{split}
			\beta_{\bm{\ell}}(\bm{t}) &= \sum_{\bm{k} \in L_{\bm{\ell}}} \bar{f}(\bm{k}) {}_{2}\wal_{\bm{k}}(\bm{t})\\
			&= \sum_{\bm{k} \in \bigotimes_{j=1}^{{d+1}} \Delta_j T_{\bm{\ell}}} \bar{f}(\bm{k}) {}_{2}\wal_{\bm{k}}(\bm{t})\\
			&= \bigotimes_{j=1}^{{d+1}} 
            \mc{D}_j
            \sum_{\bm{k} \in T_{\bm{\ell}}} \bar{f}(\bm{k}) {}_{2}\wal_{\bm{k}}(\bm{t})
		\end{split}
	\end{equation}
By Lemma~\ref{lemma:liu2024-lemma1}, we have
\begin{equation}
\label{eq:to-integral-representation}
    \sum_{\bm{k} \in T_{\bm{\ell}}} \bar{f}(\bm{k}) {}_{2}\wal_{\bm{k}}(\bm{t}) = 
\left(
\prod_{j=1}^{{d+1}} 2^{\ell_j} 
\right)
\int_{\cap_{j=1}^{d+1}
            \left\{
            \left\lfloor y_j 2^{l_j}  \right\rfloor = \left\lfloor {t}_j 2^{l_j}  \right\rfloor  
            \right\}
            } 
            f(\bm{y})\dd \bm{y}.
            \end{equation}

In the one-dimensional case (i.e., $d=0$), when $\ell > 0$, 
we have
	\begin{equation}
		\label{eq:beta_ell_1d}
		\begin{split}
		\beta_{\ell}(t) &= 2^{\ell} \int_{\left\lfloor y2^{\ell} \right\rfloor = \left\lfloor t2^{\ell} \right\rfloor} f(y) \dd y - 2^{\ell-1} \int_{\left\lfloor y2^{\ell-1} \right\rfloor = \left\lfloor t2^{\ell-1} \right\rfloor} f(y) \dd y\\
		&= \begin{cases}
			2^{\ell-1} \left( \int_{\left\lfloor y2^{\ell} \right\rfloor = \left\lfloor t2^{\ell} \right\rfloor} f(y) \dd y - \int_{\left\lfloor y2^{\ell} \right\rfloor = \left\lfloor t2^{\ell}+1 \right\rfloor} f(y) \dd y\right), & \text{if}\ \left\lfloor t2^{\ell} \right\rfloor = 2\left\lfloor t2^{\ell-1} \right\rfloor \\
			2^{\ell-1} \left( \int_{\left\lfloor y2^{\ell} \right\rfloor = \left\lfloor t2^{\ell} \right\rfloor} f(y) \dd y - \int_{\left\lfloor y2^{\ell} \right\rfloor = \left\lfloor t2^{\ell}-1 \right\rfloor} f(y) \dd y\right), & \text{otherwise},
		\end{cases}\\
		&= 2^{\ell-1} \left( \int_{\left\lfloor y2^{\ell} \right\rfloor = \left\lfloor t2^{\ell} \right\rfloor} f(y)  \dd y - \int_{\left\lfloor y2^{\ell} \right\rfloor = 2\left\lfloor t2^{\ell-1} \right\rfloor + \xi_{\ell}}  f(y)  \dd y \right),
		\end{split}
	\end{equation}
	where $\xi_{\ell}(t) = \left\lfloor t2^{\ell} \right\rfloor - 2\left\lfloor t2^{\ell-1} \right\rfloor + 1 \textrm{ mod } 2$.
When ${\ell} = 0$, $\beta_{0} = \int_{0}^{1} f(y)\dd y$.

To see what is going on, notice that $\left\lfloor y2^{\ell} \right\rfloor = \left\lfloor t2^{\ell} \right\rfloor$
means that $y$ and $t$ fall into the same interval of the form $[\frac{k}{2^{\ell}},\frac{k+1}{2^{\ell}})$, or equivalently, they agree on the first $\ell$ bits in their binary expansions. Similarly, $\left\lfloor y2^{\ell -1} \right\rfloor = \left\lfloor t2^{\ell -1} \right\rfloor$
means that $y$ and $t$ fall into the same interval of the form $[\frac{k'}{2^{\ell -1}},\frac{k'+1}{2^{\ell -1}})$.
The condition $\left\lfloor t2^{\ell} \right\rfloor = 2\left\lfloor t2^{\ell-1} \right\rfloor$ holds when the $\ell$-th binary digit of $t$ is zero. The two different cases (the $\ell$-th binary digit of $t$ being 0 or otherwise 1) tell us which neighboring `dyadic intervals' we should subtracting in the two integrals. Finally, $\left\lfloor t2^{\ell} \right\rfloor - 2\left\lfloor t2^{\ell-1} \right\rfloor$ is the $\ell$-th binary digit of $t$, and $\xi_\ell(t)$ `flips' that bit --- it is 1 if ``bit = 0" and 0 if ``bit = 1".

\begin{lemma}
    If $\bm{\ell} \in \mathbb{N}^{d+1}$ (i.e., $l_j \ge 1$ for $j \in [d+1]$), we have
    \begin{multline}
	\label{eq}
	\beta_{\bm{\ell}}(\bm{t}) = 
        \left(
        \prod_{j \in {[d+1]}} 2^{l_j - 1} 
        \right)
        \left( \sum_{{v}\subseteq {[d+1]}} (-1)^{\abs{{v}}} 
        \int_{\cap_{j \in {[d+1]}} 
        \left\{
        \left\lfloor {y}_j2^{l_j} \right\rfloor = 2\left\lfloor {t}_j2^{l_j-1} \right\rfloor + \xi_{\bm{\ell}, j} 
        \right\}
        } f(\bm{y})
        \dd \bm{y}
        \right),
	\end{multline}
where $\xi_{\bm{\ell}, j} = \left\lfloor t2^{l_j} \right\rfloor - 2\left\lfloor t2^{l_j-1} \right\rfloor + \mathbbm{1}_{j \in {v}} \textrm{ mod } 2$, and $|v|$ denotes the cardinality of the set $v$.
\end{lemma}
\begin{proof}
By definition, when $l_j \geq 1$, each operation 
$
\mc{D}_j
        \sum_{\bm{k} \in T_{\bm{\ell}}} \bar{f}(\bm{k}) {}_{2}\wal_{\bm{k}}(\bm{t})$
        yields two terms: 
    $\sum_{\bm{k} \in T_{\bm{\ell}}} \bar{f}(\bm{k}) {}_{2}\wal_{\bm{k}}(\bm{t})$ 
    and 
    $- \sum_{\bm{k} \in T_{\bm{\ell} - \bm{e}_j}} \bar{f}(\bm{k}) {}_{2}\wal_{\bm{k}}(\bm{t})$.
Therefore,
    after applying $\bigotimes_{j=1}^{d+1} \mc{D}_j$, there are $2^{d+1}$ terms with the form 
    $$
        \sum_{\bm{k} \in T_{\bm{\ell}}} \bar{f}(\bm{k}) {}_{2}\wal_{\bm{k}}(\bm{t}),
    $$
    each of which can be converted to the integral representation as in~\eqref{eq:to-integral-representation}. 
    For the general term
    $
    \sum_{\bm{k} \in T_{\bm{\ell}^\prime }} \bar{f}(\bm{k}) {}_{2}\wal_{\bm{k}}(\bm{t})$, define the set $H_{\boldell^\prime} := \{ j \in [d+1]: l_j^\prime \neq l_j \}$.
    In the integral 
    $$ \int_{\cap_{j=1}^{d+1}
            \left\{
            \left\lfloor y_j 2^{l_j^\prime}  \right\rfloor = \left\lfloor \bm{t}_j 2^{l_j^\prime}  \right\rfloor  
            \right\}
            } 
            f(\bm{y})\dd \bm{y}, $$
     the integration region $\MHlprime$ can be written as  
    $$\begin{aligned}
    &\MHlprime := \cap_{j=1}^{d+1}
            \left\{
            \left\lfloor y_j 2^{l_j^\prime}  \right\rfloor = \left\lfloor {t}_j 2^{l_j^\prime}  \right\rfloor  
            \right\} \\
    &=\{ 
    \by \in [0,1)^{d+1}: \floor{y_j 2^{l_j -1} } = \floor{t_j 2^{l_j -1} }, \text{ if } j \in \Hlprime; \quad \floor{y_j 2^{l_j } } = \floor{t_j 2^{l_j} } \text{ otherwise}
    \}.
    \end{aligned}
    $$
    The sign before the integral 
    is $(-1)^{|\Hlprime|}.$
    The factor before the integral is 
    $\prod_{j=1}^{d+1} 2^{l_j^\prime}$, which can be written as 
    $ (\prod_{j=1}^{d+1} 2^{l_j -1}) \cdot 
    2^{d+1-|H_{\boldell^\prime}|}.$
    Thus, the coefficient before the integral
    $ \int_{ M(H_{\bl^{\prime}}) } f(\by)\dd \by $
    is 
    $ (\prod_{j=1}^{d+1} 2^{l_j -1}) \cdot (-1)^{|H_{\boldell^\prime}|}
    2^{d+1-|H_{\boldell^\prime}|}$. 
    
    For $v \subseteq [d+1]$, define
    $$
    J_v:= \{ 
    \by \in [0,1)^{d+1}: \floor{y_j 2^{l_j} } = 2\floor{t_j 2^{l_j -1} }+ \xi_{\bl,j},  
    \; j \in [d+1]
    \}.
    $$
    It means $y_j$ and $t_j$ agree on the first $l_j-1$ bits of their binary expansions, but differ on the $l_j$-th digit, if $j \in v$; and $y_j$ and $t_j$ agree on the first $l_j$ bits of their binary expansions if $j \notin v$.
    Note that $\MHlprime$ can be divided into sets of the form $J_v$. 
    To prove the lemma, it suffices to show that the coefficients before $\int_{J_v} f(\by) \dd \by$ for $v \subseteq [d+1]$ is $(\prod_{j \in [d+1]} 2^{l_j-1}) (-1)^{|v|}$.
    By the symmetry of the $(d+1)$ dimensions, we only need to consider $v=\{ m+1,\ldots, d+1\}$
    where $m \in \{0,1,2,\ldots,d\}$.
    Note that $\Hlprime$ needs to include $m+1, m+2,\ldots, d+1$, since the $l_j$-th bits of $y_j$ and $t_j$ differ, for $j \in v$. And the set $\Hlprime$ can include $r$ elements of the set $\{1,2,\ldots,m\}$ where $r$ ranges from $0$ to $m$, in which case we have $|\Hlprime|=d-m+1+r$. Since there are $\binom{m}{r}$ combinations where $\Hlprime$ includes $r$ elements of $\{1,2,\ldots,m\}$, the coefficient before $\int_{J_v} f(\by) \dd \by$ is 
    $$
    (\prod_{j \in [d+1]} 2^{l_j-1})
    \sum_{r=0}^m (-1)^{|\Hlprime|}
    2^{d+1-|\Hlprime|} \binom{m}{r}
    $$
    which is exactly $(\prod_{j \in [d+1]} 2^{l_j-1}) (-1)^{|v|}$.
\end{proof}

Below we consider the case where some elements of $\bl$ may be $0$.
\begin{lemma}
When
   $\bm{\ell} \in \mathbb{N}_0^{d+1}$,
    let $u:=\{j \in \{ 1,2,\ldots, d+1 \}: l_j \neq 0 \}$, and $-u:=\{ 1,2,\ldots, d+1 \} \setminus u$.
If $u$ is nonempty, 
    we have
	\begin{multline}
		\label{eq:beta_ell_partial_s}
		\beta_{\bm{\ell}}(\bm{t}) = 
        \left(
        \prod_{j \in {u}} 2^{l_j - 1} 
        \right)\\
        \left( \sum_{{v}\subseteq {u}} (-1)^{\abs{{v}}} \int_{[0, 1]^{\abs{-{u}}}} \int_{\cap_{j \in {u}} 
        \left\{
        \left\lfloor \bm{y}_j2^{l_j} \right\rfloor = 2\left\lfloor \bm{t}_j2^{l_j-1} \right\rfloor + \xi_{\bm{\ell}, j} 
        \right\}
        } f(\bm{y})\dd \bm{y}_{{u}}\dd \bm{y}_{-{u}} \right),
	\end{multline}
where $\xi_{\bm{\ell}, j} = \left\lfloor t2^{l_j} \right\rfloor - 2\left\lfloor t2^{l_j-1} \right\rfloor + \mathbbm{1}_{j \in {v}} \textrm{ mod } 2$. 
\end{lemma}
\begin{proof}
    Note that the operation $\bigotimes_{j=1}^{d+1} \mc{D}_j$ is equivalent to the operation
    $\bigotimes_{j \in u} \mc{D}_j$.
    After applying the operation
    $\bigotimes_{j \in u} \mc{D}_j$,
    there are $2^{\cardu}$ terms with the form 
    $$
        \sum_{\bm{k} \in T_{\bm{\ell}}} \bar{f}(\bm{k}) {}_{2}\wal_{\bm{k}}(\bm{t}).
    $$
    For the general term
    $
    \sum_{\bm{k} \in T_{\bm{\ell}^\prime }} \bar{f}(\bm{k}) {}_{2}\wal_{\bm{k}}(\bm{t})$,  define the set $H_{\boldell^\prime} := \{ j \in u: l_j^\prime \neq l_j \}$. 
    By~\eqref{eq:to-integral-representation}, 
    the term $
    \sum_{\bm{k} \in T_{\bm{\ell}^\prime }} \bar{f}(\bm{k}) {}_{2}\wal_{\bm{k}}(\bm{t})$
    can be converted to an integral format as below, with a multiplying constant before it:
    \begin{equation}\label{eq:f-integral}
        \int_{\MprimeHlprime}f(\by) \dd \by, 
    \end{equation}
    where the integration region 
    $\MprimeHlprime$
    is 
    $$
    \left\{ 
    \by \in [0,1)^{d+1}: 
    \floor{y_j 2^{l_j^\prime} } = \floor{t_j 2^{l_j^\prime} }, \text{ if } j \in u; 
    \quad 
    y_j \in [0,1) \text{ if } j \notin u
    \right\}.
    $$
    The integral~\eqref{eq:f-integral} can be written as 
    \begin{equation}\label{eq:f-integral2}
    \int_{[0,1)^{|-u|}} 
    \int_{\MHlprime}
    f(\by) \dd \by_u \dd \by_{-u}
    \end{equation}
    by Fubini's Theorem,
    where 
    $\MHlprime$ is defined as 
    $$
    \left\{ \by_u \in [0,1)^{|u|}:
    \floor{y_j 2^{l_j^\prime}}
    =
    \floor{t_j 2^{l_j^\prime}}, 
    \, j \in u
    \right\}.
    $$
    The sign before the integral~\eqref{eq:f-integral2}
    is $(-1)^{|\Hlprime|}.$
    The factor before the integral is 
    $\prod_{j \in u} 2^{l_j^\prime}$, which can be written as 
    $ (\prod_{j\in u} 2^{l_j -1}) \cdot 
    2^{\cardu-|H_{\boldell^\prime}|}.$
    Thus, the coefficient before the integral
    is 
    $ (\prod_{j\in u} 2^{l_j -1}) \cdot (-1)^{|H_{\boldell^\prime}|}
    2^{\cardu-|H_{\boldell^\prime}|}$. 
    
    For $v \subseteq u$, define
    $$
    J_v:= \{ 
    \by \in [0,1)^{\cardu}: \floor{y_j 2^{l_j} } = 2\floor{t_j 2^{l_j -1} }+ \xi_{\bl,j},  
    \; j \in u
    \}.
    $$
    It means $y_j$ and $t_j$ agree on the first $l_j-1$ bits of their binary expansions, but differ on the $l_j$-th digit, if $j \in v$; and $y_j$ and $t_j$ agree on the first $l_j$ bits of their binary expansions if $j \notin v$.
    Note that $\MHlprime$ can be divided into sets of the form $J_v$. 
    To prove the lemma, it suffices to show that the coefficients before 
    $\int_{[0,1)^{|-u|}} \int_{J_v} f(\by) \dd \by_u \dd \by_{-u}$ 
    for $v \subseteq u$ 
    is $(\prod_{j \in u} 2^{l_j-1}) (-1)^{|v|}$.
    By the symmetry of the $\cardu$ dimensions, we only need to consider $v=\{ m+1,m+2,\ldots, \cardu \}$
    where $m \in \{0,1,2,\ldots,\cardu-1\}$.
    Note that $\Hlprime$ need to include $m+1, m+2,\ldots, \cardu$, since the $l_j$-th bits of $y_j$ and $t_j$ differ, for $j \in v$. And the set $\Hlprime$ can include $r$ elements of the set $\{1,2,\ldots,m\}$ where $r$ ranges from $0$ to $m$, in which case we have 
    $|\Hlprime|=\cardu-m+r$. 
    Since there are $\binom{m}{r}$ combinations where $\Hlprime$ includes $r$ elements of $\{1,2,\ldots,m\}$, the coefficient before 
     $\int_{[0,1)^{|-u|}} \int_{J_v} f(\by) \dd \by_u \dd \by_{-u}$
    is 
    $$
    (\prod_{j \in u} 2^{l_j-1})
    \sum_{r=0}^m (-1)^{|\Hlprime|}
    2^{\cardu-|\Hlprime|} \binom{m}{r},
    $$
    which is exactly $(\prod_{j \in u} 2^{l_j-1}) (-1)^{|v|}$.
\end{proof}

Next, we bound 
$ 
 \sup_{x,x' \in \mc{X}}
 \sigma^2_{\bm{\ell}}
$
in the one-dimensional case. 
	\begin{lemma}
    Assume integrands $f_{x,x'} \in L^2([0, 1])$ satisfy Condition~\ref{assumption:owen_boundary_growth_condition}
  with $A=0$.
  Let $\sigma^2_{l}:=\sigma^2_{l}(f_{x,x'})$ be as defined in \eqref{eq:sigma2_l}.
Then we have	
        \begin{equation}			\label{eq:sigma2_ell_decay_1d}
			\sup_{x,x'\in \mc{X}} \sigma^2_{\ell} \leq C^2 \pi^2 2^{-l-1}.
		\end{equation}
	\end{lemma}

\begin{proof}
For notational simplicity, we omit the subscripts of $f_{x,x'}$ and write it as $f$ in the following.
	For $\ell > 0$, by~\eqref{eq:beta_ell_1d}, 
	we have
    \begin{equation*}
		\begin{split}
			\sigma^2_{\ell} &= \int_{[0,1]} \beta^2_{\ell}(x)dx\\
			&= 2^{2\ell-2} \int_{[0,1]} \left( \int_{\left\lfloor y2^{\ell} \right\rfloor = \left\lfloor t2^{\ell} \right\rfloor} f(y)  \dd y - \int_{\left\lfloor y2^{\ell} \right\rfloor = 2\left\lfloor t2^{\ell-1} \right\rfloor + \xi_{\ell}}  f(y)  \dd y \right)^2 \dd t\\
			&= 2^{2 \ell-2}  \sum_{k = 0}^{2^{\ell}-1} 2^{-\ell} \left( \int_{\left\lfloor y2^{\ell} \right\rfloor = k} f(y) \dd y - \int_{\left\lfloor y2^{\ell} \right\rfloor =  k+1 } f(y) \dd y\right)^2 \cdot \mathbbm{1}_{k \textrm{ mod } 2 = 0}\\
			&+ 2^{2 \ell-2}  \sum_{k = 0}^{2^{\ell}-1} 2^{-\ell} \left( \int_{\left\lfloor y2^{\ell} \right\rfloor = k} f(y) \dd y - \int_{\left\lfloor y2^{\ell} \right\rfloor =  k - 1 } f(y) \dd y\right)^2 \cdot \mathbbm{1}_{k \textrm{ mod } 2 = 1} \\
			&= 2^{2 \ell-2} \sum_{k = 0}^{2^{\ell-1}-1} 2^{-\ell} \cdot 2 \left( \int_{\left\lfloor y2^{\ell} \right\rfloor = 2k} f(y) \dd y - \int_{\left\lfloor y2^{\ell} \right\rfloor =  2k+1 } f(y) \dd y\right)^2 \\
			&= 2^{\ell-1} \sum_{k = 0}^{2^{\ell-1}-1} \left( \int_{\left\lfloor y2^{\ell} \right\rfloor =  2k } f(y) - f(y+2^{-\ell}) \dd y\right)^2 \\
			&\leq 2^{\ell-1} \sum_{k = 0}^{2^{\ell-1}-1} \left( \int_{\left\lfloor y2^{\ell} \right\rfloor =  2k } \abs{f(y) - f(y+2^{-\ell})} \dd y\right)^2. 
		\end{split}
	\end{equation*}
 For a given $y_0 \in (0, \frac{1}{2} - 2^{-\ell})$, 
 when $\left\lfloor y_0 2^{\ell} \right\rfloor =  2k$,
$k$ belongs to $\{0,\dotsc, 2^{\ell - 2} - 1\}$. In such a case, we have
	\begin{equation}
		\label{eq:sobol_first_half_k}
		\begin{split}
			\sup_{x,x'\in \mc{X}}
   \abs{f(y_0) - f(y_0+2^{-\ell})} 
   &= 
   \sup_{x,x'\in \mc{X}}
   \abs{\int_{y_0}^{y_0+2^{-\ell}} \frac{\partial f}{\partial y} \dd y}\\
			&\leq
   \sup_{x,x'\in \mc{X}}
   \int_{y_0}^{y_0+2^{-\ell}} \abs{\frac{\partial f}{\partial y} } \dd y\\
   &\leq C
   \int_{y_0}^{y_0+2^{-\ell}} y^{-1-A} \dd y.
		\end{split}
	\end{equation}
	We consider the symmetricity of the boundary growth condition in $[0, 1]$. For $y_0 \in (0, \frac{1}{2} - 2^{\ell})$, $1 - 2^{-\ell} - y_0 \in (\frac{1}{2}, 1-2^{-\ell})$, which corresponds to $k \in\{ 2^{\ell - 2}, \dotsc, 2^{\ell-1} - 1\}$ when $\left\lfloor (1 - 2^{-\ell} - y_0) 2^{\ell} \right\rfloor =  2k$. In such a case, we have 
	\begin{equation}
		\label{eq:sobol_second_half_k}
		\begin{split}
        \sup_{x,x'\in \mc{X}}
			\abs{f(1 - 2^{-\ell} - y_0) - f(1 - y_0)} 
   &=
   \sup_{x,x'\in \mc{X}}
   \abs{\int_{1 - 2^{-\ell} - y_0}^{1 - y_0} \frac{\partial f}{\partial y} \dd y}\\
			&\leq 
   \sup_{x,x'\in \mc{X}}
   \int_{1 - 2^{-\ell} - y_0}^{1 - y_0} \abs{\frac{\partial f}{\partial y} } \dd y\\
   &\leq C
   \int_{1 - 2^{-\ell} - y_0}^{1 - y_0} (1-y)^{-1-A} \dd y\\
			&= C \int_{y_0}^{y_0+2^{-\ell}} y^{-1-A} \dd y.
		\end{split}
	\end{equation}
	The equations~\eqref{eq:sobol_first_half_k} and~\eqref{eq:sobol_second_half_k} cover all the cases for $k = 0, \dotsc, 2^{\ell - 1} - 1$. Thus we have
	\begin{equation}
		\label{eq:sigma_2_ell_bound_difference}
		\begin{split}
		\sigma^2_{\ell} &\leq 2^{\ell-1} \sum_{k = 0}^{2^{\ell-1}-1} \left( \int_{\left\lfloor y2^{\ell} \right\rfloor =  2k } \abs{f(y) - f(y+2^{-\ell})} \dd y\right)^2\\
		&= 2^{\ell} \sum_{k = 0}^{2^{\ell-2}-1} \left( \int_{\left\lfloor y2^{\ell} \right\rfloor =  2k } \abs{f(y) - f(y+2^{-\ell})} \dd y\right)^2.
		\end{split} 
	\end{equation}

 Therefore,
	\begin{equation*}
		\begin{split}
			\sup_{x,x'\in \mc{X}} \sigma^2_{\ell} 
   &\leq 
   \sup_{x,x'\in \mc{X}}
   2^{\ell} \sum_{k=0}^{2^{\ell - 2} - 1} \left( \int_{\left\lfloor y2^{\ell} \right\rfloor =  2k } \abs{f(y) - f(y+2^{-\ell})} \dd y\right)^2\\
            & \leq 2^l \sum_{k=0}^{2^{l-2}-1} 
            \left( \int_{\lfloor y2^l \rfloor=2k} C \int_y^{y+2^{-l}} \frac{1}{t}\dd t \dd y\right)^2\\
            &=
            C^2 2^{\ell} \sum_{k=0}^{2^{\ell - 2} - 1} \left( \int_{2k\cdot 2^{-\ell}}^{(2k+1)2^{-\ell}} \log(y+2^{-\ell}) - \log(y) \dd y \right)^2\\
		\end{split}
	\end{equation*}
    Let

     $\varTheta(x) = x\log x - x$. 
Then
\begin{equation*}
    \sup_{x,x'\in \mc{X}} \sigma^2_{\ell} 
   \leq 
    C^2 2^{\ell} 
    \sum_{k=0}^{2^{\ell - 2} - 1} 
    \left( 
    \varTheta\left((2k+2)2^{-\ell}\right) + \varTheta(2k \cdot 2^{-\ell}) - 2\varTheta((2k+1) \cdot 2^{-\ell}) 
    \right)^2. 
\end{equation*}
Note that when $k=0$, $\varTheta\left((2k+2)2^{-\ell}\right) + \varTheta(2k \cdot 2^{-\ell}) - 2\varTheta((2k+1) \cdot 2^{-\ell})$ is equal to $(\log 2) \cdot  2^{-l+1}$.
 For a general $k$, we have Taylor expansions:
	\begin{equation*}
		\begin{split}
			\varTheta\left((2k+2)2^{-\ell}\right) &= \varTheta\left((2k+1)2^{-\ell}\right) + 2^{-\ell} \varTheta^{\prime}\left((2k+1)2^{-\ell}\right)\\
			&\hspace*{12em} + \cdots + (2^{-\ell})^m \frac{\varTheta^{(m)} ((2k+1)2^{-\ell}) }{m!} + \cdots\\
			\varTheta\left((2k)2^{-\ell}\right) &= \varTheta\left((2k+1)2^{-\ell}\right) - 2^{-\ell} \varTheta^{\prime}\left((2k+1)2^{-\ell}\right)\\
			&\hspace*{12em} + \cdots + (-2^{-\ell})^m \frac{\varTheta^{(m)} ((2k+1)2^{-\ell}) }{m!} + \cdots
		\end{split}
	\end{equation*}
    Therefore,
	\begin{equation*}
		\begin{split}
			&\varTheta\left((2k+2)2^{-\ell}\right) + \varTheta(2k \cdot 2^{-\ell}) - 2\varTheta((2k+1) \cdot 2^{-\ell})\\
			&= \sum_{m=1}^{\infty} 2 (2^{-\ell})^{2m} \frac{\varTheta^{(2m)} ((2k+1)2^{-\ell}) }{(2m)!}\\
			&= 2 \sum_{m=1}^{\infty} (2^{-\ell})^{2m} \frac{\varTheta^{(2m)} ((2k+1)2^{-\ell}) }{(2m)!}\\
			&= 2 \sum_{m=1}^{\infty} (2^{-\ell})^{2m} \frac{(2m-2)! ((2k+1)2^{-\ell})^{-2m+1} }{(2m)!}\\
			&= 2^{-\ell + 1} \sum_{m=1}^{\infty} \frac{(2m-2)! (2k+1)^{-2m+1} }{(2m)!}. 
		\end{split}
	\end{equation*}	
When $k \geq 1$, with the fact that $\frac{(2k+1)^{-2m+1}}{(2m)(2m-1)}\leq \frac{(2k+1)^{-2m+1}}{12}$ for $m\geq 2$, we have 
\begin{align*}
 \sum_{m=1}^{\infty} \frac{(2k+1)^{-2m+1}}{(2m)(2m-1)} 
 &\leq \frac{1}{2} (2k+1)^{-1} + 
 \frac{1}{12}\cdot 
 \frac{(2k+1)^{-3}}{1 - (2k+1)^{-2}} \\
& \leq \frac{1}{2} (2k+1)^{-1} + \frac{1}{12}\cdot 
\frac{9}{8} (2k+1)^{-3} \\
& \leq \frac{1}{2} (2k+1)^{-1} + \frac{1}{12} \cdot \frac{1}{8} (2k+1)^{-1} \\
& \leq (2k+1)^{-1}. 
\end{align*}
Therefore, 
 \begin{align*}
 \sup_{x,x'\in \mc{X}}
 \sigma_l^2 
    &\leq C^2 \cdot 2^l \sum_{k=0}^{2^{l-2} - 1} 2^{-2l+2} (2k+1)^{-2} \leq \frac{C^2 \pi^2}{8} 2^{-l+2},
\end{align*}
where the last inequality follows from the fact that $\sum_{k=0}^\infty (2k+1)^{-2}={\pi^2}/{8}$.

\end{proof}

Now we consider the multi-dimensional case. To this end, we first show the following lemma \ref{lemma:second-equality}. Note that, $f(\by_u;\by_{-u})$ denotes a function of $y_j, \, j \in u$, with $y_j,\,j \in -u$ being fixed.
Let $f(\by_v:(\by+2^{-\bl})_{u\setminus v};\by_{-u})$
denote a function $f(\by^\prime)$ where 
$y^\prime_j,\, j\in -u$ is fixed, 
$y^\prime_j=y_j$ if $j \in v$, and 
$y^\prime_j=y_j+2^{-l_j}$ if $j \in u\setminus v$.
\begin{lemma}
    \label{lemma:second-equality}
Assume $f_{\x,\x'} \in L^2([0, 1]^{d+1})$,    
    and  $\mc{X}$ is a compact set.  
    Let $\sigma^2_{\bl}:=\sigma^2_{\bl}(f_{\x,\x'})$ be as defined in \eqref{eq:sigma2_l}.
    Then
\begin{equation*}
\begin{split}
    \sup_{\x,\x' \in \mc{X}} \sigma_{\bm{\ell}}^2 
    \leq 
    2^{{d+1}+\normone{\bl}}
    \sup_{\x,\x'\in \mc{X}}
  &\sum_{\substack{
    \bm{k}_{{u}} \in \mathbb{N}_0^{\abs{{u}}} \\ 
    { \bm{k}_{{u}}\leq 
    2^{ {\bm{\ell}}_{{u}} -1 } - 1 } 
    }}
   \biggl(
   \int_{[0, 1]^{\abs{-{u}}} } \int_{\cap_{j \in {u}} \left\lfloor y_j 2^{l_j} \right\rfloor = 2k_j} \\
 & \sum_{v \subseteq {u}} (-1)^{\abs{v}} f\left(\bm{y}_{v}:\left(\bm{y} + 2^{-\bm{\ell}} \right)_{{u}-v} 
  ;
  \bm{y}_{-{u}} \right) \dd \bm{y}_{{u}} \dd \bm{y}_{-{u}} 
  \biggr)^2.
  \end{split}
\end{equation*}
\end{lemma}
\begin{proof}
    By Lemma~\ref{lemma:dick2010_1332},
    we have
    \begin{equation*}
        \sup_{\x,\x'\in \mc{X}}
		\sigma^2_{\bm{\ell}} 
  =
  \sup_{\x,\x'\in \mc{X}}
  \int_{[0,1]^{d+1}} \abs{\beta_{\bm{\ell}}(\bm{t})}^2 {\rm d}\bm{t}.\\
    \end{equation*}
    From \eqref{eq:beta_ell_partial_s}, we have 
\begin{equation}\label{eq:integral-t}
    \begin{split}
        \sup_{\x,\x' \in \mc{X}} \sigma_{\bl}^2
        = \sup_{\x,\x' \in \mc{X}} \prod_{j \in u}2^{2l_j-2} 
        \int_{[0,1]^{d+1}}
        \bigg(
        \sum_{v \subseteq u} (-1)^{|v|}   
        \int_{[0,1]^{|-u|}} 
        & \int_{\cap_{j \in u}\{\lfloor y_j 2^{l_j}\rfloor= 
        2\lfloor t_j 2^{l_j-1} \rfloor
        +\xi_{lj}
        }\\
        & f(\by)\dd \by_u \dd \by_{-u}
        \bigg)^2
        \dd \bt.
        \end{split}
    \end{equation}
    For any $j \in u$, let $k_j$ be even integer in $[0,2^{l_j}-2]$. Define
    \begin{equation*}
        I_j=
        \begin{cases} 
      1, & \text{if } 
      \frac{k_j}{2^{l_j}} \leq t_j < \frac{k_j+1}{2^{l_j}} \\
      2, & \text{if }  
        \frac{k_j+1}{2^{l_j}} \leq t_j < \frac{k_j+2}{2^{l_j}}.
        \end{cases}
    \end{equation*}  
    Recall that for $\bm{\ell} \in \mathbb{N}_0^{d+1}$,
    $u:=\{j \in [{d+1}]: l_j \neq 0 \}$. 
    Let $u = \{i_1, \ldots, i_{|u|}\}$.
    Define
    \begin{equation*}
        \mcZ(j,v,I_j,k_j):=
        \left\{
        \by_u \in [0,1)^{\cardu}:
            \floor{y_j 2^{l_j}}=
             \begin{cases}
            k_j+1, &\text{ if} \; j \in v, I_j=1
        \text{  or} \; j \notin v, I_j=2,\\
        k_j, &\text{ otherwise.}
        \end{cases}
        \right\}.
    \end{equation*}
Note that after fixing $k_j, I_j, j \in u$, the value of 
$$
\left( 
        \sum_{v \subseteq u} (-1)^{|v|}   
        \int_{[0,1]^{|-u|}} 
        \int_{\cap_{j \in u}
        \mcZ(j,v, I_j,k_j)
        }
        f(\by) \dd \by_u \dd \by_{-u}
        \right)^2
$$
is fixed.
So in \eqref{eq:integral-t} the integral with respect to  $\bt$ can be written as sum over $k_j, I_j, j \in u.$
    Specifically, we have
    \begin{align}
\sup_{\x,\x' \in \mc{X}} \sigma_{\bl}^2
        = \sup_{\x,\x' \in \mc{X}}
        \left(
        \prod_{j \in u} 
        2^{2l_j-2} 
        \right)
        & \left(
        \prod_{j \in u} 2^{-l_j} 
        \right)
        \sum_{k_{i_{1}}} 
        \cdots 
        \sum_{k_{i_{|u|}}}
        \sum_{I_{i_{1}} \in \{1,2\}}
        \cdots
        \sum_{I_{i_{|u|}} \in \{1,2\}} \nonumber \\
     &
     \left( 
        \sum_{v \subseteq u} (-1)^{|v|}   
        \int_{[0,1]^{|-u|}} 
        \int_{\cap_{j \in u}
        \mcZ(j,v,I_j,k_j)
        }
        f(\by)\dd \by_u \dd \by_{-u}
        \right)^2,
\end{align}
where for any $j \in u$, the sum of $k_j$ is over all even integers in $[0,2^{l_j}-2]$.

Note that, fixing $k_{i_1},\ldots, k_{i_{|u|}},I_{i_2},\ldots, I_{i_{|u|}} $,
we can construct a bijection between subsets of $u$ as follows: for any $v \subseteq u$, if $i_1 \in v$, let $\tilde{v}=v\setminus \{i_1\}$; 
if $i_1 \notin v$, let $\tilde{v}=v \cup \{i_1\}$.
Therefore, the values of 
\begin{equation*}
    \left( 
        \sum_{v \subseteq u} (-1)^{|v|}   
        \int_{[0,1]^{|-u|}} 
        \int_{\cap_{j \in u}
        \mcZ(j,v,I_j,k_j)
        }
        f(\by)\dd \by_u \dd \by_{-u}
        \right)^2
\end{equation*}
with $I_{i_1}=1$ and $I_{i_1}=2$ are the same.
For notational simplicity, let 
\begin{equation*}
    \mc{A}(\by_u):= 
    \left\{ 
    \by_u:
    \lfloor y_j 2^{l_j} \rfloor=
    \begin{cases} 
      k_j, & \text{if } j \in v,  j\in u \\
      k_j+1, & \text{if } j \notin v, j \in u 
   \end{cases}
    \right\}.
\end{equation*}
By applying the same technique to $i_2,\ldots i_{|u|}$, we derive that

\begin{equation}
    \label{eq:}
    \begin{split}
    \sup_{\x,\x' \in \mc{X}} \sigma_{\bl}^2 
    &=
    \sup_{\x,\x' \in \mc{X}}
    \Bigl(
      \prod_{j \in u}
      2^{l_j - 2}
    \Bigr)
    \; 2^{|u|}\\
    &\qquad \sum_{k_{i_1}} \cdots \sum_{k_{i_{|u|}}}
    \Bigl(
      \sum_{v \subseteq u} (-1)^{|v|}
      \int_{[0,1]^{|-u|}}
      \int_{\mc{A}(\by_u)}
      f(\by_u; \by_{-u})
      \,\dd \by_u \,\dd \by_{-u}
    \Bigr)^2,
     \end{split}
\end{equation}
where each sum of $k_j, j\in u$ is over all even integers in $[0,2^{l_j}-2].$
Given the fact that $\cardu \leq d+1$, the sum $\sum_{k_{i_1}}\cdots \sum_{k_{i_{\cardu}}}$ can be written in a compact form
$
\sum_{\substack{
            \bm{k}_{{u}} \in \mathbb{N}_0^{|{u}|} \\
            \bm{k}_{{u}} \leq 2^{\bm{\ell}_{{u}}} - 2 \\
            k_j \text{ even}
        }}
$, and with the change of variable
$
y_j^{\prime}=y_j-2^{-l_j},\, j \in u\setminus v,
$
we have
\begin{equation*}
\begin{split}
    \sup_{\x,\x' \in \mc{X}} \sigma_{\bl}^2 
    &=
    \sup_{\x,\x' \in \mc{X}}
    \Bigl(
      \prod_{j \in u}
      2^{l_j - 2}
    \Bigr)
    \; 2^{|u|}\\
    & \qquad \sum_{k_{i_1}} \cdots \sum_{k_{i_{|u|}}}
    \Bigl(
      \sum_{v \subseteq u} (-1)^{|v|}
      \int_{[0,1]^{|-u|}}
      \int_{\mc{A}(\by_u)}
      f(\by_u; \by_{-u})
      \,\dd \by_u \,\dd \by_{-u}
    \Bigr)^2
    \\[6pt]
    & \leq
    \begin{aligned}[t]
      & \sup_{\x,\x' \in \mc{X}} \; 2^{d+1} 
        \Bigl(
          \prod_{j \in u} 2^{l_j - 2}
        \Bigr)
        \sum_{\substack{
            \bm{k}_{{u}} \in \mathbb{N}_0^{|{u}|} \\
            \bm{k}_{{u}} \leq 2^{\bm{\ell}_{{u}}} - 2 \\
            k_j \text{ even}
        }}
      \\[4pt]
      & 
        \Bigl(
          \sum_{v \subseteq u}
          (-1)^{|v|}
          \int_{[0,1]^{|-u|}}
          \int_{\cap_{j \in u} \{\lfloor y_j 2^{l_j}\rfloor = k_j\}}
          f\bigl(\by_v : (\by + 2^{-\bl})_{u \setminus v};\,\by_{-u}\bigr)
          \,\dd \by_u \,\dd \by_{-u}
        \Bigr)^2
    \end{aligned}
    \\[8pt]
    & \leq
    \begin{aligned}[t]
      & 2^{(d+1) + \normone{\bl}}
        \;\sup_{\x,\x' \in \mc{X}}
        \sum_{\substack{
           \bm{k}_{{u}} \in \mathbb{N}_0^{|{u}|} \\
           \bm{k}_{{u}}\le 
           2^{ {\bm{\ell}}_{{u}}-1 } - 1
        }}
      \\[4pt]
      & 
        \Bigl(
          \int_{[0,1]^{|\!-\!{u}|}}
          \int_{\cap_{j \in {u}} \{\lfloor y_j 2^{l_j}\rfloor = 2k_j\}}
          \sum_{v \subseteq {u}} (-1)^{|v|}
          f\bigl(\bm{y}_v : (\bm{y}+2^{-\bm{\ell}})_{{u}\setminus v};\,\bm{y}_{-{u}}\bigr)
          \,\dd \bm{y}_{{u}} \,\dd \bm{y}_{-{u}}
        \Bigr)^2.
    \end{aligned}
\end{split}
\end{equation*}
\end{proof}
    
With Lemma~\ref{lemma:second-equality}, we show the following result for the multidimensional case.
\begin{lemma}	\label{lemma:variance_sobol_multi_d}
	Assume integrands $f_{\x,\x'} \in L^2([0, 1]^{d+1})$ satisfy 
Condition~\ref{assumption:owen_boundary_growth_condition}
 with all $A_j=0$.
 Then we have
	\begin{equation*}
 \sup_{\x,\x'\in \mc{X}}
		\sigma^2_{\bm{\ell}} 
  \leq C^2 \cdot 2^{-\| \boldell \|_1 + 5(d+1)}.
	\end{equation*}
\end{lemma}
\begin{proof}
For notational simplicity, we omit the subscripts of $f_{\x,\x'}$ and write it as $f$ below.
By lemma \ref{lemma:second-equality},
we have

\begin{equation}
\label{eq:sigma2_ell_multi_dimension}
\begin{split}
  \sup_{\x,\x'\in \mc{X}} \sigma^2_{\bm{\ell}} 
  & \leq
  \begin{aligned}[t]
    & 2^{{d+1}+\normone{\boldell}}
      \sup_{\x,\x'\in \mc{X}}
      \sum_{\substack{
        \bm{k}_{{u}} \in \mathbb{N}_0^{\abs{{u}}} \\
        \bm{k}_{{u}}\leq 2^{ {\bm{\ell}}_{{u}}-1 } - 1
      }}
    \\[-3pt]
    & 
      \Bigl(
        \int_{[0, 1]^{\abs{-{u}}}}
        \int_{\cap_{j \in {u}}
          \bigl\{\lfloor y_j 2^{l_j} \rfloor = 2k_j\bigr\}}
        \sum_{v \subseteq {u}} (-1)^{\abs{v}}
        f\bigl(\bm{y}_{v} : (\bm{y}+2^{-\bm{\ell}} )_{{u}-v};\, \bm{y}_{-{u}}\bigr)
        \,\dd \bm{y}_{{u}}
        \,\dd \bm{y}_{-{u}}
      \Bigr)^2
  \end{aligned}
  \\[6pt]
  & =
  \begin{aligned}[t]
    & 2^{{d+1}+\normone{\boldell}}
      \sup_{\x,\x'\in \mc{X}}
      \sum_{\substack{
        \bm{k}_{{u}} \in \mathbb{N}_0^{\abs{{u}}} \\
        \bm{k}_{{u}}\leq 2^{{\bm{\ell}}_{{u}}-1 } - 1
      }}
    \\[-3pt]
    & \quad
      \Bigl(
        \int_{(E_{\bm{\ell}, 2\bm{k}})_{{u}}}
        \sum_{{v} \subseteq {u}} (-1)^{\abs{{v}}}
        \int_{[0, 1]^{\abs{-{u}}}}
          f\bigl(
            \bm{y}_{{v}} : (\bm{y}+2^{-\bm{\ell}})_{{u}-{v}};
            \,\bm{y}_{-{u}}
          \bigr)
        \,\dd \bm{y}_{-{u}}
        \,\dd \bm{y}_{{u}}
      \Bigr)^2
  \end{aligned}
  \\[6pt]
  & =
  2^{{d+1}+\normone{\boldell}}
  \sup_{\x,\x'\in \mc{X}}\\
 &\qquad  \sum_{\substack{\bm{k}_{{u}} \in \mathbb{N}_0^{\abs{{u}}} \\
                  \bm{k}_{{u}}\leq 2^{{\bm{\ell}}_{{u}}-1 } - 1}}
  \Bigl(
    \int_{{(E_{\bm{\ell}, 2\bm{k}})_{{u}}}}
    \int_{[\bm{y}_u, \,\bm{y}_u + 2^{-\bm{\ell}_{{u}}}]}
    \partial^{{u}}
    \int_{[0, 1]^{\abs{-{u}}}}
      f(\bm{y}_0;\,\bm{y}_{-{u}})
    \,\dd \bm{y}_{-{u}}
    \,\dd \bm{y}_0
    \,\dd \bm{y}_u
  \Bigr)^2
  \\[6pt]
  & \leq
  C^2 \cdot 2^{{d+1}+\normone{\boldell}}\\
  &\qquad \sum_{\substack{\bm{k}_{{u}} \in \mathbb{N}_0^{\abs{{u}}} \\
                  \bm{k}_{{u}}\leq 2^{{\bm{\ell}}_{{u}}-1 } - 1}}
  \Bigl(
    \int_{{(E_{\bm{\ell}, 2\bm{k}})_{{u}}}}
    \int_{[\bm{y}_u, \,\bm{y}_u + 2^{-\bm{\ell}_{{u}}}]}
      \prod_{j\in {u}} 
      \min(y_{0j},1-y_{0j})
      ^{-A_j - 1}
    \,\dd \bm{y}_0
    \,\dd \bm{y}_u
  \Bigr)^2.
\end{split}
\end{equation}

When there is an $l_j = 1$ with $A_j = 0$, then

\begin{align*}
\int_0^{\frac{1}{2}} \int_{y_j}^{y_j + \frac{1}{2}} 
\min(t,1-t)^{-1}
\dd t \dd y_j &= \int_0^{\frac{1}{2}} \left( \int_{y_j}^{\frac{1}{2}} \frac{1}{t} \dd t + \int_{\frac{1}{2}}^{y_j+\frac{1}{2}} \frac{1}{1-t} \dd t \right) \dd y_j = 1.
\end{align*}
When all $l_j > 1$ for $j \in {u}$, 
by symmetry, we have
\begin{equation*}
\begin{split}
    \sup_{\x,\x'\in \mc{X}}	
 \sigma^2_{\bm{\ell}} 
 &\leq 
 C^2 \cdot 2^{d+1}
 \cdot 
  2^{{d+1}+\normone{\boldell}}\\
 &\qquad \sum_{\substack{\bm{k}_{{u}} \in \mathbb{N}_0^{\abs{{u}}} \\ {\bm{k}_{{u}}\leq 
 2^{
 {\bm{\ell}}_{{u}} - 2
 } - 1} }} \left(  \int_{(E_{\bm{\ell}, 2\bm{k}})_{{u}}}
	\int_{[\bm{y}_u, \bm{y}_u + 2^{-\bm{\ell}_{{u}}}] } \prod_{j\in {u}} 
    \min(y_{0j},1-y_{0j})^{-1}
    \dd \bm{y}_0 \dd \bm{y}_u  \right)^2\\
&=C^2 \cdot 2^{d+1}
 \cdot 
 2^{{d+1}+\normone{\boldell}}
 \sum_{\substack{\bm{k}_{{u}} \in \mathbb{N}_0^{\abs{{u}}} \\ {\bm{k}_{{u}}\leq 
 2^{
 {\bm{\ell}}_{{u}} - 2
 } - 1} }} \left(  \int_{(E_{\bm{\ell}, 2\bm{k}})_{{u}}}
    \prod_{j\in {u}} \int_{y_j}^{y_j + 2^{-l_j}} 
    {y}_{0 j}^{ - 1} 
    \dd {y}_{0j}
    \dd \bm{y}_u  \right)^2\\
&=C^2 \cdot 2^{d+1}
 \cdot 
 2^{{d+1}+\normone{\boldell}}\\
 & \qquad  \sum_{\substack{\bm{k}_{{u}} \in \mathbb{N}_0^{\abs{{u}}} \\ {\bm{k}_{{u}}\leq 
 2^{
 {\bm{\ell}}_{{u}} - 2
 } - 1} }} \left(  \int_{(E_{\bm{\ell}, 2\bm{k}})_{{u}}}
    \prod_{j \in u} 
        \left( 
        \log({y}_{j} + 2^{-l_j}) - \log({y}_{j})
        \right)
    \dd \bm{y}_u  \right)^2.\\
\end{split}
\end{equation*}
Therefore, 
\begin{align*}
\sup_{\x,\x'\in \mc{X}}
\sigma_{\bl}^2 
&\leq C^2 \cdot 
2^{2(d+1)+\normone{\boldell}}\\
&\qquad \sum_{\substack{\bk_u \in \mathbb{N}_0^{|u|} \\ \bk_u \leq 
2^{\bl_u-2} - 1}
} 
\left( \prod_{j \in u} \int_{ \left[ \frac{2{k_j}}{2^{l_j}}, \frac{2{k_j + 1}}{2^{l_j}} \right)} 
\left( \log \left( y_j + 2^{-l_j} \right) - \log y_j \right) \dd y_j \right)^2 \\
&\leq C^2 
\cdot 
2^{2(d+1)+\normone{\boldell}}
\sum_{
\substack{
\bk_u \in \mathbb{N}_0^{|u|} \\ 
\bk_u \leq 
2^{\bl_u-2} - 1
}
} 
\left[ \prod_{j \in u} \left( 2^{-l_j + 1} \left( 2k_j + 1 \right)^{-1} \right) \right]^2 \\
&\leq 
C^2 \cdot 
2^{-\normone{\boldell}+4(d+1)}
\sum_{
\substack{
\bk_u \in \mathbb{N}_0^{|u|} \\ 
\bk_u \leq 
2^{\bl_u-2} - 1}
}
\prod_{j \in u} \left( 2k_j + 1 \right)^{-2} \\
& \leq 
C^2 \cdot 
2^{-\normone{\boldell}+4(d+1)}
\prod_{j \in u}
\sum_{k_j=0}^{2^{l_j-2}-1}
(2k_j+1)^{-2}\\
& \leq 
C^2 \cdot 
2^{-\normone{\boldell}+4(d+1)}
\left(
\frac{\pi^2}{8}
\right)^{d+1}\\
& \leq 
C^2 \cdot 
2^{-\normone{\boldell}+5(d+1)}.
\end{align*}
\end{proof}

Now we are ready to prove Theorem~\ref{thm:sup_meansquare}. 

\begin{proof}

Let $f_{\x,\x'}(\bomega)=\psi(\x,\bomega)\psi(\x',\bomega)$ with $\sigma^2_{\bm{\ell}}:=\sigma^2_{\bm{\ell}}(f)$.

If $K(\cdot, \cdot)$ is a shift-invariant kernel satisfying Condition \ref{cond1}, let $\bm{w}=(\bm{t},b)$, then by \citet[Appendix B.1]{huangquasi},
the function $f_{\x,\x'}$ 
can be re-written as
$$\begin{aligned}
f_{\x,\x'}(\mathbf{t},b) &= \cos\left( (\x-\x')^\top \BPhi^{-1}(\mathbf{t})  \right) -  \cos\left( (\x+\x')^\top \BPhi^{-1}(\mathbf{t}) + 4\pi b \right).
\end{aligned} $$
Let $D=\max_{\x,\mathbf{y}\in \mathcal{X},i\in\{1,\ldots,d\}} \{ |x_i-{y}_i|,|x_i+{y}_i|\}$. Then for any non-empty set $u\subset \{1,\ldots,d+1\}$ and $(\mathbf{t},b)\in(0,1)^{d+1}$,
$$\left|\partial^u f_{\x,\x'}
(\mathbf{t},b) \right| \leq 4\pi  D^{|u\backslash\{d+1\}|} \prod_{i\in u\backslash\{d+1\}}\frac{{\rm d} }{{\rm d}t_i}\Phi^{-1}_i(t_i).$$
By Condition \ref{cond1}, $\frac{{\rm d} }{{\rm d}t}\Phi^{-1}_i(t)\leq \frac{C_i}{\min(t,1-t)} $ for some constant $C_i>0$ and all $ t\in (0,1)$. Therefore, the Condition \ref{assumption:owen_boundary_growth_condition} is satisfied with $C = 4\pi  D^{|u\backslash\{d+1\}|} \prod_{i\in u\backslash\{d+1\}} C_i
$ and all $A_j=0$.

Recall that the first $M=2^m$ points of a scrambled Sobol' $(t,s)$-sequence $(m\ge t \ge 0)$ is used.
By Lemma \ref{lemma:2} and \ref{lemma:variance_sobol_multi_d},  we have 
\begin{align*}
\sup_{\x,\x' \in \mc{X}} \mathbb{E} \left[ |K_M(\mathbf{x}, \mathbf{x}') - K(\mathbf{x}, \mathbf{x}')|^2 \right] 
 & \leq 2^{-m + t + {d+1}} 
 \sup_{\x,\x'\in \mc{X}}
 \sum_{\substack{\bl \in \mathbb{N}_0^{d+1} \\ \| \bl \|_1 > m - t}} \sigma_{\bl}^2(f) \\
& \leq C^2 \cdot 2^{-m + t + 6(d+1)} \sum_{\substack{\bl \in \mathbb{N}_0^{d+1} \\ \|\bl\|_1 > m - t}} 
2^{-\normone{\boldell}} \\
& = C^2 \cdot 2^{-m + t + 6(d+1)} \sum_{k=m-t+1}^{\infty} 2^{-k} \binom{k + {d+1} - 1}{{d+1} - 1}.
\end{align*}
By Lemma \ref{lemma:1}, we have 
\begin{align*}
 \sum_{k=m-t+1}^{\infty} \left( \frac{1}{2} \right)^k \binom{k + {d+1} - 1}{{d+1} - 1} \leq 2^{-(m - t + 1)+{d+1}} \binom{m - t + {d+1}}{{d+1} - 1}.
 \end{align*}
Note that when $m \ge 4,$ we have 
$\binom{m-t+d+1}{d} \le 2m^d/d!.$
Therefore,
\begin{align*}
 \sup_{\x,\x' \in \mc{X}} \mathbb{E} \left[ |K_M(\mathbf{x}, \mathbf{x}') - K(\mathbf{x}, \mathbf{x}')|^2 \right] 
 &\leq \frac{C^2}{M^2} 2^{2t+7(d+1)-1} \binom{m - t + {d+1}}{{d+1} - 1} \\
& \leq 
 C^2 \cdot \frac{2^{2t+7(d+1)}}{d!}  \frac{\log^d_2 M  }{M^2}.
\end{align*}
\end{proof}

\subsection{Proof of Theorem~\ref{thm:qmc1_deterministiccase_kernelapprox_error}}
\label{pf:thm:qmc1_deterministiccase_kernelapprox_error}
Recall that we use the first $M=b^{m}$ points of a scrambled $(t, d+1)$ sequence in base $b$.
When $m \geq t$, the first $M=b^{m}$ points of a $(t, d+1)$ sequence is a $(t, m, d+1)$ net, which remains a $(t, m, d+1)$ net with probability 1 after scrambling 
\citep[Proposition 17.2]{owen2023practical}.
It has the following property: for any subinterval of $[0,1)^{d+1}$ of the form
$
\prod_{j=1}^{d+1}\left[\frac{c_{j}}{b^{k_{j}}}, \frac{c_{j}+1}{b^{k_{j}}}\right)
$
with $k_{j} \geq 0$ and $0 \leq c_{j}<b^{k_{j}}$, if it is of volume $b^{t-m}$, then it contains exactly $b^{t}$ points of the sequence.

Now, we consider $f:[0,1]^{d+1} \rightarrow \mathbb{R}, f(\bomega)=\psi(\mathbf{x}, \bomega) \psi\left(\mathbf{x}^{\prime}, \bomega\right) \leq \kappa^{2}$. Let $\tilde{f}_{M}$ be the low variation function that coincides with $f$ on a ``large set" $K_{M}=\left[\varepsilon_{M}, 1-\varepsilon_{M}\right]^{d+1}$ defined in \citet[Appendix B.1]{huangquasi}. We have
\begin{equation*}
    \begin{split}
\left|\int_{[0,1]^{d+1}} f(\mathbf{x}) \mathrm{d} \mathbf{x}-\frac{1}{M} \sum_{i=1}^{M} f\left(\mathbf{h}_{i}\right)\right| \leq \int_{[0,1]^{d+1}}
&\left|f(\mathbf{x})-\tilde{f}_{M}(\mathbf{x})\right| \mathrm{d} \mathbf{x}\\
&+\mathcal{D}^{*}\left(\left\{\mathbf{h}_{i}\right\}_{i=1}^{M}\right) V_{\mathrm{HK}}\left(\tilde{f}_{M}\right)\\
& +\frac{1}{M} \sum_{i=1}^{M}\left|\tilde{f}_{M}\left(\mathbf{h}_{i}\right)-f_{M}\left(\mathbf{h}_{i}\right)\right|.
   \end{split}
\end{equation*}
By 
\citet[Inequality B.4]{huangquasi},
$$
V_{\mathrm{HK}}\left(\tilde{f}_{M}\right) \leq 2 B\left(1-2 \log 2-2 \log \varepsilon_{M}\right)^{d},
$$
where $B=4\pi D^{|u\setminus \{d+1\}|} \prod_{i \in u \setminus \{d+1\}}C_i$, 
and 
$D=\max_{\x,\y \in \mc{X}, \, i \in \{1,\ldots,d\}} \{ \abs{x_i-y_i}, \abs{x_i+y_i}  \}$.

By \citet[Inequality B.6]{huangquasi},
$$
\int_{[0,1]^{d+1}}\left|f(\mathbf{x})-\tilde{f}_{M}(\mathbf{x})\right| \mathrm{d} \mathbf{x} \leq 3 \cdot 2^{d-1} B \varepsilon_{M}\left(2+(2-\log 2) d-d \log \varepsilon_{M}\right).
$$

Since $\tilde{f}_{M}$ coincides with $f$ on $K_{M}=\left[\varepsilon_{M}, 1-\varepsilon_{M}\right]^{d+1}$, the region where $\tilde{f}_{M}$ differs from $f$ can be covered by $2(d+1)$ subintervals:

$$
\begin{array}{ll}
{\left[0, \varepsilon_{M}\right] \times[0,1] \times \cdots \times[0,1],} & {\left[1-\varepsilon_{M}, 1\right] \times[0,1] \times \cdots \times[0,1]} \\
{[0,1] \times\left[0, \varepsilon_{M}\right] \times \cdots \times[0,1],} & {[0,1] \times\left[1-\varepsilon_{M}, 1\right] \times \cdots \times[0,1]} \\
{[0,1] \times \cdots \times[0,1] \times\left[0, \varepsilon_{M}\right],} & {[0,1] \times \cdots \times[0,1] \times\left[1-\varepsilon_{M}, 1\right]}
\end{array}
$$

Let $\varepsilon_{M}=\frac{1}{b^{m-t}}=b^{t} / M$. Then each of these intervals contains exactly $b^{t}$ points, and thus there are at most $2(d+1) b^{t}$ points in the union of these intervals. 
Note that
$
\tilde{f}(\mathbf{x})=f\left(\operatorname{Proj}_{\left[\varepsilon_{M}, 1-\varepsilon_{M}\right]}\left(x_{1}\right), \ldots, \operatorname{Proj}_{\left[\varepsilon_{M}, 1-\varepsilon_{M}\right]}\left(x_{d}\right)\right)
$, where $\operatorname{Proj}_{\left[\varepsilon_{M}, 1-\varepsilon_{M}\right]}\left(x\right)$ is the projection of $x$ onto $\left[\varepsilon_{M}, 1-\varepsilon_{M}\right]$.
Therefore,
$$
\frac{1}{M} \sum_{i=1}^{M}\left|\tilde{f}_{M}\left(\mathbf{h}_{i}\right)-f_{M}\left(\mathbf{h}_{i}\right)\right| \leq \frac{2(d+1) b^{t}}{M} \cdot 2 \kappa
.$$

The star discrepancy of a $(t, m, d+1)$-net in base $b$ satisfies $\mathcal{D}^{*}\left(\left\{\mathbf{h}_{i}\right\}_{i=1}^{M}\right) \leq C \cdot \frac{(\log M)^{d}}{M}$ for some constant $C$
\citep[Theorem 4.10]{niederreiter1992random}.
Combining the bounds above, we have
\begin{equation*}
    \left|\int_{[0,1]^{d+1}} f(\mathbf{x}) \mathrm{d} \mathbf{x}-\frac{1}{M} \sum_{i=1}^{M} f\left(\mathbf{h}_{i}\right)\right| \leq C^{\prime} \cdot \frac{\log ^{2 d} M}{M}.
\end{equation*}

\subsection{Proof of Theorem~\ref{Thm:RQMCF-KRR}}
\label{pf:thm:rqmcf-krr}
Given the deterministic error bound of RQMC features shown in Theorem~\ref{thm:qmc1_deterministiccase_kernelapprox_error} and \ref{thm:QMC2}, the same proof as in \citet[Appendix C]{huangquasi}  applies.

\section{Supplementary Technical Lemmas}

\begin{lemma}\citep{dick2010digital}
\label{lemma:dick2010_1332}
    With the notations in Section~\ref{pf:thm:sup_meansquare},
    we have
    \begin{equation*}
        \sup_{\x,\x'\in \mc{X}}
		\sigma^2_{\bm{\ell}} 
  =
  \sup_{\x,\x'\in \mc{X}}
  \int_{[0,1]^{d+1}} \abs{\beta_{\bm{\ell}}(\bm{t})}^2 {\rm d}\bm{t}.\\
    \end{equation*}
\end{lemma}

\begin{lemma}\citep{liu2024randomized}
\label{lemma:liu2024-lemma1}
	Given the function $f \in L^2([0, 1]^{d+1})$ and the index set $T_{\bm{\ell}}$ as defined in~\eqref{eq:set_T_ell}, we have the expression for the Walsh series in base 2 in $T_{\bm{\ell}}$ as
	\begin{equation*}
		\label{eq:walsh_series_T_ell}
		\sum_{\bm{k} \in T_{\bm{\ell}}} \bar{f}(\bm{k}) {}_{2}\wal_{\bm{k}}(\bm{t}) = \prod_{j=1}^{{d+1}} 2^{\ell_j} \int_{\cap_{j=1}^{d+1} \left\lfloor y_j 2^{l_j}  \right\rfloor = \left\lfloor {t}_j 2^{l_j}  \right\rfloor} f(\bm{y})\dd \bm{y}.
	\end{equation*}
	\label{lemma:walsh_series_T_ell}
\end{lemma}

\begin{lemma}\citep{dick2010digital}
    \label{lemma:1}
    For any real number $b > 1$ and any $k, t_0 \in \mathbb{N}$, we have
\[
\sum_{t=t_0}^{\infty} b^{-t} \binom{t+k-1}{k-1} \leq b^{-t_0} \binom{t_0+k-1}{k-1} \left( 1 - \frac{1}{b} \right)^{-k}.
\]
\end{lemma}

\begin{lemma}\citep{dick2010digital}
\label{lemma:2}
    Let $f \in L_2([0, 1]^{d+1})$ and let $\hat{I}(f)$ be the RQMC estimator using the scrambled Sobol' sequence. 
    Then
\begin{equation*}
    \text{Var}[\hat{I}(f)] \leq b^{-m + t + {d+1}} \sum_{\substack{
    \boldell \in \mathbb{N}_0^{d+1} \\ 
     \normone{\boldell} > m - t}} 
    \sigma_{\boldell}^2(f).
\end{equation*}
\end{lemma}

\section{Additional Simulation Results}
\label{sec:additional-simu}
In this section, 
we present experimental findings for KRR with $r=0.5$ case.

Following the same procedure as before, the training and test datasets are generated from $Y = f(\mathbf{X}) + \varepsilon$, where $f$ is the regression function, $\mathbf{X} \sim \mathrm{Unif}[0,1]^d$, and $\varepsilon \sim \mathcal{N}(0,1)$. We focus on the Gaussian kernel $K(\mathbf{x}, \mathbf{x}') = \exp\!\bigl(-\tfrac{1}{2\sigma^2}\|\mathbf{x}-\mathbf{x}'\|^2\bigr)$, where the bandwidth $\sigma$ is chosen as the median of $\|\mathbf{X} - \mathbf{X}'\|$ computed numerically, with $\mathbf{X}, \mathbf{X}'$ drawn i.i.d.\ from $\mathrm{Unif}[0,1]^d$.

The range of $L^r$ coincides with $\mathcal{H}$ when $r=0.5$. Hence, we set $\tilde{f}(\mathbf{x}) = K\bigl(\tfrac{1}{3}\mathbf{1}_d, \mathbf{x}\bigr) + K\bigl(\tfrac{2}{3}\mathbf{1}_d, \mathbf{x}\bigr)$, ensuring that $\tilde{f} \in \mathrm{ran}(L^r)$. To control the signal-to-noise ratio, we let $f(\mathbf{x}) = C_{\tilde{f}}\tilde{f}(\mathbf{x})$, where $C_{\tilde{f}}$ is chosen so that $\mathbb{E}[f(\mathbf{X})] = 5$. The kernel ridge regularization parameter is $\lambda = 0.25\,n^{-\tfrac{1}{2r+1}}$.

We plot the test MSE against the number of random features for exact KRR, RF-KRR, QMCF-KRR and RQMCF-KRR in Figure~\ref{fig:r-0.5}. For each dimension $d$, we generate $10^6$ test points and keep them fixed. We then conduct 1000 trials of training samples of size $10^4$. For each trial, we fit the kernel ridge regressor and record its test error. The MSE (solid lines) is the average over these 1000 trials, and we additionally provide confidence bands based on the 25\% and 75\% quantiles of the errors.

Empirically, the results for $r=0.5$ exhibit patterns similar to those observed for $r=1$ in Section~\ref{sec:simu-krr}, again showcasing the strong performance of RQMC-based features.

\clearpage
\begin{figure}[htbp]
\thispagestyle{empty}
  \centering
  \includegraphics[width=\textwidth, height=\textheight, keepaspectratio]{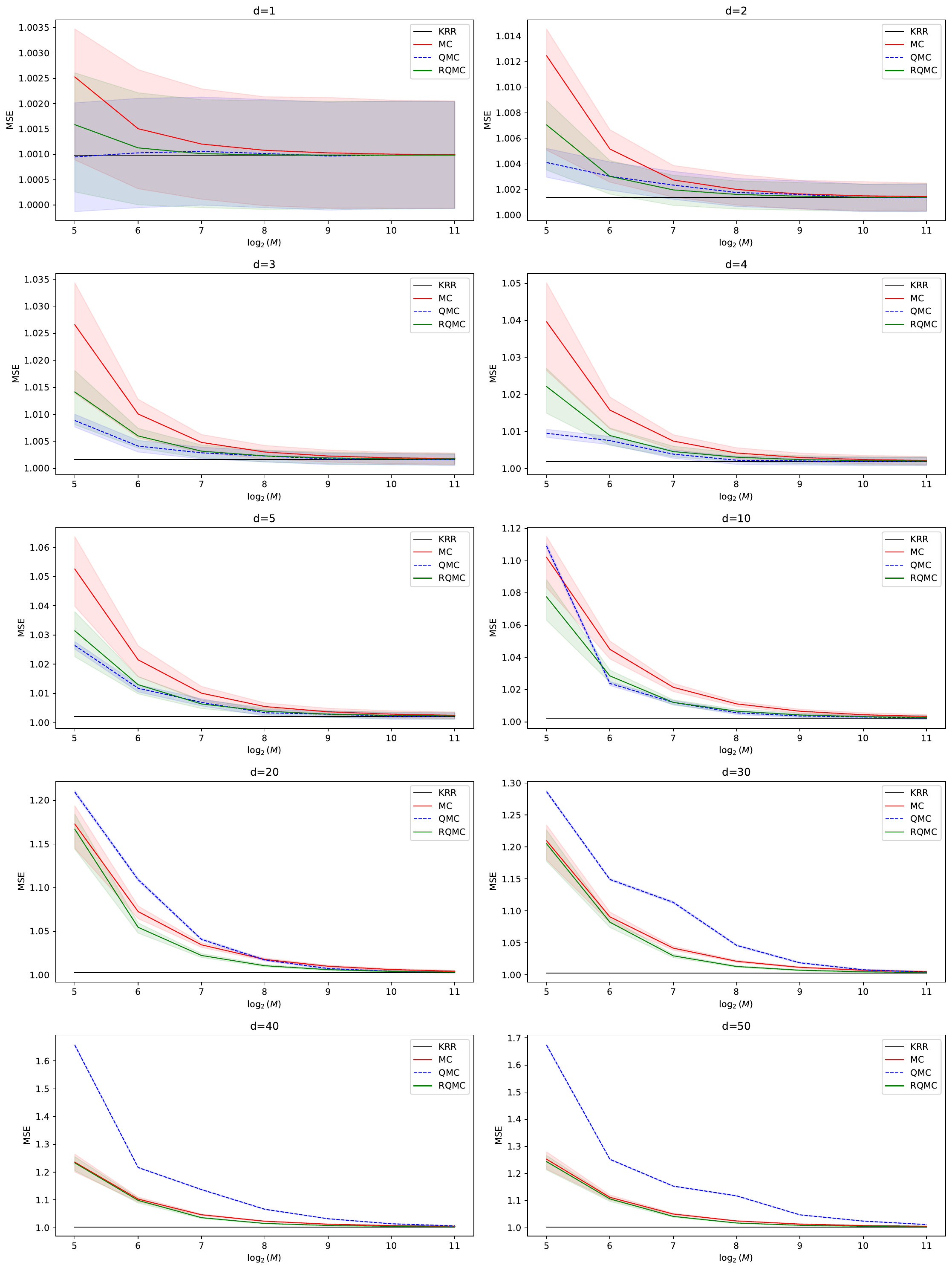}
  \caption{The test MSE against the number of random features ($r=0.5$), for exact KRR, RF-KRR, QMCF-KRR and RQMCF-KRR.}
  \label{fig:r-0.5}
\end{figure}
\clearpage

\end{document}